\RequirePackage[l2tabu, orthodox]{nag}
\documentclass[11pt,fleqn]{article}

\usepackage{microtype}
\usepackage{mathtools,amssymb,amsthm,latexsym} 
\usepackage{mathrsfs} 
\usepackage{bbm} 
\usepackage[longnamesfirst]{natbib} 
\usepackage[shortlabels]{enumitem}
\usepackage{graphicx} 
\usepackage{booktabs} 
\usepackage{threeparttable} 
\usepackage[table]{xcolor} 
\usepackage[letterpaper,top=2.9cm,bottom=2.4cm,hmargin={3cm,3cm}]{geometry}
\usepackage{verbatim} 
\usepackage{tikz} 
   \usetikzlibrary{patterns,decorations.pathreplacing}
\usepackage{hyperref} 
   \hypersetup{colorlinks=true, linkcolor=blue, citecolor=blue!50!black, urlcolor=cyan}
\usepackage[capitalize,noabbrev]{cleveref}
\usepackage{authblk}
\usepackage[titletoc,title]{appendix}

\theoremstyle{definition}

\theoremstyle{plain}
\newtheorem{Lem}{Lemma}
\newtheorem{Pro}{Proposition}
\newtheorem{Thm}{Theorem}
\newtheorem{Cor}{Corollary}
\theoremstyle{remark}

\newcommand{\Myd}{\;\mathrm{d}} 
\newcommand{\Ind}[1]{\mathbbm{1}_{#1}} 

\newcommand{\Bigo}[1]{\operatorname{O}\left(#1\right)}
\newcommand{\Bigoas}[1]{\operatorname{O}_{\text{as}}\left(#1\right)}
\newcommand{\Bigop}[1]{\operatorname{O}_{\text{p}}\left(#1\right)}

\newcommand{\Smalloas}[1]{\operatorname{o}_{\text{as}}\left(#1\right)}
\newcommand{\Smallop}[1]{\operatorname{o}_{\text{p}}\left(#1\right)}

\DeclareMathOperator*{\Exp}{\mathbb{E}} 

\title{Counterfactual Inference in\\ Duration Models with Random Censoring}
\author{Jiun-Hua Su\thanks{
This paper is the second chapter of my dissertation.
I am indebted to my advisor James Powell for valuable guidance.
I am also grateful to Bryan Graham, Aditya Guntuboyina, Michael Jansson, Demian Pouzo, and Pedro Sant'Anna for useful suggestions and discussions. Address correspondence to Jiun-Hua Su, 128 Academia Road, Section 2, Nankang, Taipei, 115 Taiwan; E-mail address: jhsu@econ.sinica.edu.tw.}}
\affil{Institute of Economics\\Academia Sinica}
\begin{document}
\maketitle

\medskip
\begin{abstract}
We propose a counterfactual Kaplan-Meier estimator that incorporates exogenous covariates and unobserved heterogeneity of unrestricted dimensionality in duration models with random censoring.
Under some regularity conditions, we establish the joint weak convergence of the proposed counterfactual estimator and the unconditional Kaplan-Meier (\citeyear{KaplanMeier1958}) estimator.
Applying the functional delta method, we make inference on the \emph{cumulative hazard policy effect}, that is, the change of duration dependence in response to a counterfactual policy.
We also evaluate the finite sample performance of the proposed counterfactual estimation method in a Monte Carlo study.

\bigskip
\noindent
\textit{Keywords}: Counterfactual policy effect, Random censoring, Duration analysis, Kaplan-Meier estimator

\medskip
\noindent
\textit{JEL Classification}: C14, C41, C24
\bigskip
\end{abstract}

\fontsize{12}{18pt}\selectfont
\newpage
\section{Introduction}\label{Introduction}
Policy evaluation is one of the important areas in social science.
Counterfactual analysis is an approach that provides policy recommendations
when a policy is not implemented yet or when a similar quasi-experiment is infeasible.
Recent studies on counterfactual analysis, for example \citet{Rothe2010} and \citet{ChernozhukovFernandez-ValEtAl2013a}, emphasize the unconditional distributional effect of an exogenous manipulation of covariates on an outcome variable of interest.
Methods in these studies are usually based on data that are completely observed; however, sampling schemes may generate incomplete data and thus restrict their applicability.
For example, duration data on unemployment spells, collected by the Current Population Survey, are commonly believed to be subject to right censoring, as explained by \citet{Kiefer1988}.

The main objectives of this paper are to estimate the unconditional distribution of a duration variable affected by a counterfactual policy that exogenously manipulates covariates, and to evaluate associated policy effects by the comparison between the counterfactual and original unconditional distribution of the duration variable.
Specifically, we consider a nonseparable model
\begin{align}\label{nonseparable}
T=\varphi(X,\varepsilon),
\end{align}
where $T$ is a nonnegative duration variable of interest, $X$ is a $d$-dimensional vector of time-invariant covariates, $\varepsilon$ is individual unobserved heterogeneity in an arbitrary measurable space of unrestricted dimensionality, and $\varphi$ is a structural function that is unknown to researchers.
In addition, the right censoring may make $T$ unobserved; instead, the observable data are the vector $X$ of covariates,
\begin{align}\label{censor}
Y=\min\{T,C\}\quad \text{and}\quad \delta=\Ind{[T\leq C]}\;,
\end{align}
where $C$ is a censoring random variable, which is only observed for censored observations, and $\Ind{[\cdot]}$ is an indicator function.
Policy makers consider the counterfactual scenario that exogenously changes $X$ to $X^{*}$ and leads to the counterfactual duration variable
\begin{align}\label{counterfactual}
T^{*}=\varphi(X^{*},\varepsilon),
\end{align}
and attempt to evaluate the policy effect
\begin{align*}\label{functional}
\nu(F_{T^{*}})-\nu(F_{T}),
\end{align*}
where $F_{T}$ and $F_{T^{*}}$ are the cumulative distribution functions (CDFs) of $T$ and $T^{*}$, respectively, and $\nu$ is some functional defined on the collection of all CDFs. 
Such an effect $\nu(F_{T^{*}})-\nu(F_{T})$ may be important in policy evaluation.
For example, although unemployment insurance benefits may smooth the income fluctuation of the unemployed, it may discourage the unemployed from searching for jobs.
Therefore, policy makers would be interested in the effect of reducing wage replacement ratio on the cumulative hazard rate of unemployment spells.
In this case, $T$ is the unemployment duration, $X$ is the wage replacement ratio, and the functional $\nu$ is a map from a CDF to its cumulative hazard function, that is, $\nu:F\mapsto
\int_{[0,\cdot]} \frac{1}{1-F^{-}}\Myd F$.\footnote{
For any c\`{a}dl\`{a}g function $F$, we write $F^{-}$ for its left-continuous version, that is,
$F^{-}(t)\equiv \lim_{s\uparrow t}F(s)$.
}


We propose a nonparametric estimation method of the unconditional CDF $F_{T^{*}}$ arising from an exogenous manipulation of covariates $X$ on the duration variable $T$.
This proposed nonparametric estimation method allows researchers to conduct a counterfactual policy analysis, rather than just a descriptive analysis, of duration data.\footnote{
\citet{Lancaster1992} and \citet{CameronTrivedi2005} indicate that the nonparametric Kaplan-Meier estimation is traditionally viewed as a descriptive analysis.
}
Specifically, we evaluate $\nu(F_{T^{*}})-\nu(F_{T})$ by replacing $F_{T^{*}}$ and $F_{T}$ with their nonparametric estimator, respectively.
On the one hand, we construct the unconditional Kaplan-Meier (\citeyear{KaplanMeier1958}) estimator $\hat{F}_{T}$.
On the other hand, under regularity conditions, the unconditional CDF of $T^{*}$ is recovered by  $F_{T^{*}}(t)=\mathbb{E}(F_{T|X}(t|X^{*}))$ where $F_{T|X}$ is the conditional CDF of $T$ given $X$; thus, we follow the analogy principle to propose a two-stage fully nonparametric estimator of the counterfactual CDF of $T^{*}$.
In particular, we first construct a variant of \citeauthor{Beran1981}'s (\citeyear{Beran1981}) conditional Kaplan-Meier estimator $\hat{F}_{T|X}$ and then take average of $\hat{F}_{T|X}$ with respect to the empirical distribution of $X^{*}$ to obtain a counterfactual estimator $\hat{F}_{T^{*}}$.
The first-stage nonparametric estimation can avoid the misspecification of the conditional CDF, which is emphasized in \citet{RotheWied2013}.
Moreover, the proposed first-stage estimator, instead of the kernel CDF estimator in \citet{Rothe2010}, is essential to avoid the estimation bias in the presence of censoring.
Indeed, our simulation experiments show that the proposed estimator $\hat{F}_{T^{*}}$, compared with Rothe's counterfactual CDF estimator, has smaller mean integrated absolute error (MIAE) and root mean integrated squared error (RMISE) when duration data are subject to censoring.

To establish the validity of the proposed approach, we show that under some regularity conditions, the vector $(\hat{F}_{T^{*}}-F_{T^{*}},\hat{F}_{T}-F_{T})^{\top}$ converges weakly to a two dimensional centered Gaussian process over a specific compact subset of $\mathbb{R}^{2}_{+}$ at the rate $\sqrt{n}$.
This convergence rate avoids the curse of dimensionality even though the first-stage estimator $\hat{F}_{T|X}$ converges at a rate less than $\sqrt{n}$.
Applying the functional delta method, we can obtain the asymptotic distribution of the counterfactual policy effect $\sqrt{n}(\nu(\hat{F}_{T^{*}})-\nu(\hat{F}_{T}))$, which allows us to evaluate the effect of counterfactual policy intervention on the cumulative hazard function.

The proposed method also complements the literature on decomposition methods.
Decomposition methods are usually used to explain the difference in unconditional distributional features of an outcome variable across two different demographic groups or time periods.
The between-group difference is usually decomposed into a structure effect and a composition effect.\footnote{
A structure effect arises because structural functions are different between two groups, and a composition effect reflects the differences in covariates between two groups.
The early development in decomposition methods is well surveyed by \citet{FortinLemieuxEtAl2011}.
Recently, \citet{Rothe2015} further investigates a detailed decomposition, which attributes the composition effect to each covariate.
}
In the presence of complete data, \citet{Rothe2010} proposes a two-stage fully nonparametric estimation of the composition effect, whereas \citet{ChernozhukovFernandez-ValEtAl2013a} develop a two-stage semiparametric estimation of this effect by either distribution regression or quantile regression.
Taking random censoring into account, \citet{Garcia-Suaza2016} studies the effect based on the proportional hazard specification.
In contrast, the method in this paper is fully nonparametric; additionally, as explained by~\citet{Rothe2010}, we can regard $X^{*}$ as observable covariates of a different group, and $\nu(F_{T^{*}})-\nu(F_{T})$ as the composition effect in the setup~(\ref{nonseparable})-(\ref{counterfactual}) of random censoring.

Throughout this paper, all random variables are defined on the same probability space $(\Omega,\mathcal{A},\mathbb{P})$.
We denote $\mathbb{D}[c_{1},c_{2}]$ and $l^{\infty}[c_{1},c_{2}]$ by the set of c\`{a}dl\`{a}g and bounded functions defined on the interval $[c_{1},c_{2}]$, respectively.
We write $\Rightarrow$ for weak convergence in a function space equipped with the uniform norm, and $a\wedge b$ for the minimum of $a$ and $b$.
We also denote the density of $X$ by $m$, and the density of $X^{*}$ by $m^{*}$.
For a generic random variable $U$, we write $F_{U}$ for the CDF of $U$, $f_{U}$ for the derivative of $F_{U}$, $F_{U|X}$ for the conditional CDF of $U$ given $X$, and $f_{U|X}(u|x)$ for the derivative with respect to $u$ of $F_{U|X}(u|x)$; additionally, let $F^{\delta}_{U}(u)=\mathbb{P}(U\leq u, \delta=1)$ and $F^{\delta}_{U|X}(u|x)=\mathbb{P}(U\leq u, \delta=1|X=x)$.
We assume that $T,C,X$, and $X^{*}$ are absolutely continuous random variables.
The absolute continuity of the duration variable $T$ is reasonable because $T$ is expected to be generated by a transition process, which is usually modeled in continuous time.
(See \citet{CameronTrivedi2005} and \citet{FlorensFougereEtAl2008} for example.)
Furthermore, we only consider absolutely continuous covariates for ease of exposition because the proposed estimation method can be revised to include discrete covariates.
Alternatively, in the case of a binary policy variable, \citet{SantAnna2016} extends the method of Kaplan-Meier integrals and studies various treatment effects when the outcome may be right censored.

The remainder of this paper is organized as follows.
Section~\ref{Model} discusses the setup of duration analysis, the objects of interest, and the counterfactual Kaplan-Meier estimator.
Section~\ref{AsympInf} shows the asymptotic theory of the proposed estimator and statistical inference on the associated policy effects.
Section~\ref{Simulation} presents the results of Monte Carlo simulation.
Section~\ref{Conclusion} concludes.
Technical proofs are deferred to Appendix.

\section{Model and estimation}\label{Model}
\subsection{Setup and objects of interest}\label{Objects}
The flexible duration model in~(\ref{nonseparable}) can avoid several types of model misspecification.
First, \citet{LuWhite2014} point out that the nonseparability of $\varepsilon$ enables treatment effect and marginal effect could depend on unobservable heterogeneity. 
In addition, the unrestricted dimensionality of $\varepsilon$ can avoid incorrect inference due to the inclusion of limited heterogeneity, as argued by \citet{BrowningCarro2007}.
\citet{HoderleinMammen2009} also indicate that $\varepsilon$ can be viewed as an element of an infinitely dimensional function space; for example, it could be individual preference for leisure in the analysis of unemployment spells.
Finally, both the marginal distribution of $\varepsilon$ and the conditional distribution of $T$ given $\varepsilon$ are unspecified to avoid inappropriate inference caused by parametric assumptions.\footnote{
\citet{Lancaster1992} documents many alternatives of parametric assumptions about the hazard function in mixture models.
Parametric specification of the heterogeneity distribution and duration dependence is, however, a well-known issue in econometrics.
See the discussion in \citet{HausmanWoutersen2014}.
}


The random censoring feature in~(\ref{censor}) is prevalent in duration analysis and usually arises because of sampling schemes, for example, a random failure to follow up an individual during the study period.
We refer readers to \citet{Moore2016} for more underlying reasons of random censoring.
In this paper, we consider the simple counterfactual scenario that policy intervention does not
affect the structural function $\varphi$ in~(\ref{counterfactual}); however, we allow a change in the censoring variable $C$ after policy intervention.\footnote{
As suggested in \citet{FortinLemieuxEtAl2011}, policy intervention may result in an alternative structural function $\varphi^{*}$ in general equilibrium.
}

A counterfactual policy that changes the duration variable from $T$ to $T^{*}$ yields the \emph{distribution policy effect}
\begin{align*}
\triangle_{F}(t)\equiv F_{T^{*}}(t)-F_{T}(t).
\end{align*}
For instance, since the shape of the distribution of unemployment spells may affect the government expenditures on unemployment insurance, the distribution policy effect matters for policy makers concerning fiscal deficits.
In addition, the literature on duration models especially emphasizes the duration dependence, that is, the shape of the hazard function.\footnote{
The hazard function of a nonnegative duration variable $T$ is defined as
\begin{align*}
\lambda_{T}(t)\equiv \lim_{u\to 0}\frac{\mathbb{P}(t\leq T<t+u|T\geq t)}{u}.
\end{align*}
}
The change of duration dependence in response to a counterfactual policy can be answered by the \emph{cumulative hazard policy effect}
\begin{align*}
\triangle_{\Lambda}(t)\equiv\Lambda_{T^{*}}(t)-\Lambda_{T}(t),
\end{align*}
where $\Lambda_{T^{*}}(t)=\int_{0}^{t}\frac{F_{T^{*}}(\!\!\Myd u)}{1-F_{T^{*}}^{-}(u)}$ and $\Lambda_{T}(t)=\int_{0}^{t}\frac{F_{T}(\!\!\Myd u)}{1-F_{T}^{-}(u)}$ are the cumulative hazard functions of $T^{*}$ and $T$, respectively.
In the case of unemployment spells, policy makers would be interested in the cumulative hazard policy effect because it would evaluate whether a counterfactual policy is beneficial for a target group, for example the long-term unemployed, to escape the unemployment trap.
Other counterfactual policy effects, such as quantile policy effect and Lorenz curve policy effect, may also be of interest. See \citet{Bhattacharya2007} and \citet{Rothe2010} for treatment of these and further examples.
When the objects of interest are the aforementioned policy effects, the identification of $\varphi$ is not necessary, as indicated by \citet{Rothe2010}; thus, we maintain the flexible specification of the structural function in~(\ref{nonseparable}).\footnote{
\citet{Matzkin2003} provides conditions such that in the absence of censoring, the structural function $\varphi$ can be identified if $\varphi(x,\varepsilon)$ is strictly increasing in unobserved scalar heterogeneity $\varepsilon$ for each $x$.
}

\subsection{Nonparametric identification and estimation}
Since a counterfactual policy effect can be generally written as $\nu(F_{T^{*}})-\nu(F_{T})$ for some specific functional $\nu$, we start by identifying the CDFs $F_{T}$ and $F_{T^{*}}$.
We first introduce the following assumptions.\\[0.2cm]
\textbf{Assumption D} (Data)
\begin{enumerate}[label=D\arabic*]
\item \label{D1} Both $\{(Y_{i},T_{i},C_{i},\delta_{i},X_{i})\}_{i=1}^{n}$ and $\{X_{j}^{*}\}_{j=1}^{n^{*}}$ are independent and identically distributed across individuals.
\item \label{D2}
(i) $\{(Y_{i},\delta_{i},X_{i})\}_{i=1}^{n}$ are observable;
(ii) $\{X_{j}^{*}\}_{j=1}^{n^{*}}$ are observable and $n^{*}=n$.
\end{enumerate}

Assumption~\ref{D1} is common in models of cross-sectional data.
Assumption~\ref{D2}(i) is also common in duration models where researchers know whether the observed duration variable is censored.
Assumption~\ref{D2}(ii) is innocuous when we consider the counterfactual policy that shifts $X$ to $X^{*}=\pi(X)$ for some measurable function $\pi$, whereas this assumption is imposed for convenience when we treat $X^{*}$ as observable covariates of a different group in the analysis of the composition effect.\\

\noindent
\textbf{Assumption I} (Identification)
\begin{enumerate}[label=I\arabic*]
\item \label{I1} $T$ and $C$ are independent.
\item \label{I2} $T$ and $C$ are conditionally independent given $X$.
\item \label{I3} $\varepsilon$ is independent of both $X$ and $X^{*}$.
\item \label{I4} The support of $X^{*}$ is a subset of the support of $X$.
\end{enumerate}

Assumptions~\ref{I1} and~\ref{I2} are commonly imposed in survival analysis,
for example \citet{Lancaster1992} and \citet{KalbfleischPrentice2002} for Assumption~\ref{I1} and \citet{Dabrowska1989}, \citet{Iglesias-PerezGonzalez-Manteiga1999}, and \citet{Gneyou2014} for Assumption~\ref{I2}.
Assumption~\ref{I1} ensures that the random censoring is non-informative; that is, $C$ does not provide any information about $T$, and vice versa.
Assumption~\ref{I2} holds if random censoring is non-informative when covariates are controlled.
Note that neither Assumption~\ref{I1} nor Assumption~\ref{I2} is stronger.\footnote{
See the examples on page 65 of \citet{Stoyanov2014}.}
If $C$ is independent of $(X,T)$, then Assumptions~\ref{I1} and~\ref{I2} are satisfied; however, these two assumptions are not sufficient for the independence between $C$ and $X$.\footnote{
The independence between $C$ and $X$ holds if Assumptions~\ref{I1} and~\ref{I2} hold and the family of distributions of $T$ given $X$ is boundedly complete; that is, for a bounded function $g$, $\Exp[g(T)|X]=0$ almost surely implies $g(T)=0$ almost surely.
See the discussion in \citet{Dawid1998}.
}
Assumptions~\ref{I3} and \ref{I4} are imposed for counterfactual analysis as in \citet{Rothe2010}.
Assumption~\ref{I3} requires that all covariates are exogenous and thus may be strong in some empirical studies.
If we are interested in the effects arising from the manipulation of some policy variables, this assumption can be weaken by conditional exogeneity of $\varepsilon$ given observable covariates.
To be precise, let $X=(X_{\text{p}},X_{\text{c}})^{\top}$ where $X_{\text{p}}$ and $X_{\text{c}}$ are the vector of policy variables and vector of covariates, respectively.
Policy intervention changes $X_{\text{p}}$ to $X^{*}_{\text{p}}$ but keeps $X_{\text{c}}$ unchanged.
Proposition~\ref{PRO1} below is still valid if Assumption~\ref{I3} is replaced with the assumption that $\varepsilon$ is independent of $(X_{\text{p}},X^{*}_{\text{p}})$ conditional on $X_{\text{c}}$.\footnote{
Similarly, Assumptions~\ref{I3} can be relaxed by the control function approach proposed by \citet{BlundellPowell2003} and \citet{ImbensNewey2009}.
The application of the control function approach is however beyond the scope this paper.
See \citet{Lee2015} for the analysis of counterfactual effects by the control function approach in the absence of random censoring.
}
Assumptions~\ref{I2} and~\ref{I3} imply that censoring occurs exogenously provided that $\varphi(x,\cdot)$ is invertible for all $x$.\footnote{
Suppose that invertibility of $\varphi(x,\cdot)$ holds for all $x$.
Assumption~\ref{I2} implies that $\varepsilon$ and $C$ are conditionally independent given $X$ by Lemmas 4.1 and 4.2 of \citet{Dawid1979}.
If moreover Assumption~\ref{I3} holds, then $\varepsilon$ is independent of $(C,X)$ by Lemma 4.2 of \citet{Dawid1979}.
}
Since nonparametric analyses of counterfactual policy effects resulting from an extrapolation of covariates may be invalid, we impose the overlap condition in Assumption~\ref{I4}.

Assumption I guarantees the identification of $(F_{T^{*}},F_{T})^{\top}$ over a subset of $\mathbb{R}_{+}^{2}$.
\citet{StuteWang1993} show that under Assumption~\ref{I1}, $F_{T}(t)$ is identified for each $t<\tau\equiv\inf\{t:F_{Y}(t)=1\}$.
Under Assumptions~\ref{I3} and~\ref{I4}, we can express $F_{T^{*}}$ as the population average, taken with respect to the distribution of $X^{*}$, of the conditional CDF of $T$ given $X$; to be precise, $F_{T^{*}}(t)=\mathbb{E}(F_{T|X}(t|X^{*}))$.
The identification of $F_{T^{*}}(t)$ thus follows that of $F_{T|X}(t|x)$, and the latter is achieved under Assumption~\ref{I2} for $(t,x)\in[0,\tau)\times\mathbb{R}^{d}$.
Alternatively, the joint CDF of $(T,X)$ can be identified on $[0,\tau)\times\mathbb{R}^{d}$ by replacing Assumption~\ref{I2} with the assumption that $\delta$ and $X$ are conditionally independent given $T$; that is, given the duration, covariates provide no further information on whether censoring occurs.\footnote{
Suppose that Assumption~\ref{I1} holds. If $\delta$ and $X$ are conditionally independent given $T$, then the joint distribution of $(T,X)$ on $[0,\tau)\times \mathbb{R}^{d}$ can be recovered by
\begin{align*}
\mathbb{P}(T\leq t, X\leq x)=
\int \Ind{[z\leq x]}\Ind{[s\leq t]}\exp\left\{\int_{0}^{s}\frac{H_{Y}^{0}(\!\!\Myd y)}{1-F_{Y}(y)}\right\}H_{YX}^{1}(\!\!\Myd s,\!\!\Myd z)
\end{align*}
where $H_{Y}^{0}(y)=\mathbb{P}(Y\leq y, \delta=0)$ and $H_{YX}^{1}(y,x)=\mathbb{P}(Y\leq y, X\leq x,\delta=1)$.
The availability of data on $\{(Y_{i},\delta_{i},X_{i})\}_{i=1}^{n}$ implies that $F_{T|X}(t|x)$ is identified for $(t,x)\in [0,\tau)\times \mathbb{R}^{d}$.
More general results are shown in Equation (1.2) of \citet{Stute1996}.
}
This assumption is imposed in recent studies on duration analysis, for example
\citeauthor{SantAnna2016} (\citeyear{SantAnna2016}, \citeyear{SantAnna2017}), \citet{Garcia-Suaza2016}, and references cited therein.
We summarize the discussion in the following proposition.

\begin{Pro}\label{PRO1}
Suppose that Assumptions~\ref{D1} and~\ref{D2} hold.
Under Assumption~\ref{I1}, $F_{T^{*}}(t)$ is identified for $t \in[0,\tau)$.
If in addition Assumption~\ref{I2} holds, then $F_{T|X}(t|x)$ is identified for $(t,x)\in [0,\tau)\times\mathbb{R}^{d}$.
Moreover, if Assumptions~\ref{I3} and~\ref{I4} are also satisfied, we have $F_{T^{*}}(t)=\mathbb{E}(F_{T|X}(t|X^{*}))$ for $t\in [0,\tau)$.\qed
\end{Pro}

Proposition~\ref{PRO1} suggests that we follow the analogy principle to construct an estimator of $F_{T^{*}}(t)$ by
\begin{equation}\label{CounterKM}
\hat{F}_{T^{*}}(t;h_{n})=\frac{1}{n}\sum_{i=1}^{n}\hat{F}_{T|X}(t|X_{i}^{*};h_{n}),
\end{equation}
where $\hat{F}_{T|X}$ is the variant of \citeauthor{Beran1981}'s (\citeyear{Beran1981}) conditional Kaplan-Meier estimator, that is
\begin{equation}\label{CKM}
\hat{F}_{T|X}(t|x;h)=1-
\prod_{j=1}^{n}\exp{\left\{-\frac{\Ind{[Y_{j}\leq t,\delta_{j}=1]}B_{n_{j}}(x;h)}{\sum_{\ell=1}^{n}\Ind{[Y_{j}\leq Y_{\ell}]}B_{n_{\ell}}(x;h)}\right\}},
\end{equation}
$\{B_{n_{\ell}}(x;h)\}_{\ell=1}^{n}$ are appropriate weights and $h$ is a tuning parameter.\footnote{
In fact, the estimator $\hat{F}_{T|X}$ is the exponential transformation of the Nalson-Aalen estimator of the cumulative hazard function of $F_{T|X}$.
Additionally,
\citeauthor{Beran1981}'s (\citeyear{Beran1981}) conditional Kaplan-Meier estimator
\begin{align*}
\hat{F}^{\text{KM}}_{T|X}(t|x;h)=1-
\prod_{j=1}^{n}\left\{1-\frac{B_{n_{j}}(x;h)}{\sum_{\ell=1}^{n}\Ind{[Y_{j}\leq Y_{\ell}]}B_{n_{\ell}}(x;h)}\right\}^{\Ind{[Y_{j}\leq t,\delta_{j}=1]}}
\end{align*}
can be viewed as the first-order Taylor series approximation of $\hat{F}_{T|X}$.
}
Different choices of weights are documented in the literature on the conditional Kaplan-Meier estimator.\footnote{
These choices include Gasser-Muller weights and Nadaraya-Watson weights.
See for example \citet{Gonzalez-ManteigaCadarso-Suarez1994}, \citet{Dabrowska1989}, and \citet{Iglesias-PerezGonzalez-Manteiga1999}.
}
In this paper, we construct the counterfactual Kaplan-Meier estimator in (\ref{CounterKM}) based on the Nadaraya-Watson weights
\begin{displaymath}
B_{n_{\ell}}(x;h_{n})=\frac{K(\frac{x-X_{\ell}}{h_{n}})}{\sum_{i=1}^{n}K(\frac{x-X_{i}}{h_{n}})},\quad\quad \ell=1,2,\cdots,n ,
\end{displaymath}
for some kernel function $K$ and bandwidth $h_{n}$.
For ease of notation, we suppress the dependence on $h_{n}$ for $\hat{F}_{T^{*}}$ and $\hat{F}_{T|X}$ hereafter.
To estimate $F_{T}(t)$, we adopt the unconditional Kaplan-Meier (\citeyear{KaplanMeier1958}) estimator \begin{equation}\label{KM}
\hat{F}_{T}(t)=1-
\prod_{j=1}^{n}\left(\frac{n-j}{n-j+1}\right)^{\Ind{[Y_{(j)}\leq t,\delta_{(j)}=1]}}
\end{equation}
where $\{(Y_{(j)},\delta_{(j)})\}_{j=1}^{n}$ are the $n$ pairs of observations ordered on the
order statistics of $\{Y_{i}\}_{i=1}^{n}$.

\section{Asymptotic theory}\label{AsympInf}
\subsection{Representations}\label{Linear}
Asymptotic properties of the unconditional Kaplan-Meier estimator $\hat{F}_{T}$ in (\ref{KM}) have been studied extensively in survival analysis. One attractive feature is that $\hat{F}_{T}(t)-F_{T}(t)$ can be approximated by an average of independent and identically distributed random variables with mean zero.\footnote{
Another appealing feature is the strong approximation for $\sqrt{n}(\hat{F}_{T}-F_{T})$ by a sequence of Gaussian processes.
See for example \citet{BurkeCsoergoEtAl1988} and \citet{MajorRejto1988}.
}
We state this representation in the following proposition for completeness.

\begin{Pro}\label{PRO2}
Under Assumptions~\ref{D1}-\ref{D2} and~\ref{I1}, for any $\zeta<\tau=\inf\{t:F_{Y}(t)=1\}$ and $t\in[0,\zeta]$,
\begin{align*}
\sqrt{n}\left(\hat{F}_{T}(t)-F_{T}(t)\right)
=\frac{1}{\sqrt{n}}\sum_{i=1}^{n}\xi(Y_{i},\delta_{i};t)+R_{n}(t),
\end{align*}
where
\begin{align*}
\xi(y,\delta;t)=\left[1-F_{T}(t)\right]\left[\frac{\Ind{[y\leq t, \delta=1]}}{1-F_{Y}(y)}-\int_{0}^{y\wedge t}\frac{F^{\delta}_{Y}(\!\!\Myd u)}{(1-F_{Y}(u))^{2}}\right],
\end{align*}
and
\begin{align*}
\sup_{t\in [0,\zeta]}|R_{n}(t)|=\Smallop{1}.
\end{align*}
\qed
\end{Pro}

The influence function of $\hat{F}_{T}$ in Proposition~\ref{PRO2} is centered at zero.
Moreover, the precise rate of approximation error $\sup_{t\in [0,\zeta]}|R_{n}(t)|$ differs if different assumptions about the data generating process are imposed (cf. \citet{LoSingh1986}, \citet{Cai1998}, \citet{ChenLo1997}).

In addition to the representation of $\hat{F}_{T}$, a similar representation of $\hat{F}_{T|X}$ in~(\ref{CKM}) has been established in the case of a univariate covariate by \citet{Iglesias-PerezGonzalez-Manteiga1999}, and further extended to the case of multivariate covariates and dependent data by \citet{LiangUna-AlvarezEtAl2012}.
Since the counterfactual Kaplan-Meier estimator $\hat{F}_{T^{*}}$ in (\ref{CounterKM}) is constructed by taking average of $\hat{F}_{T|X}$ with respect to the empirical distribution of $X^{*}$, applying the representation of $\hat{F}_{T|X}$ allows us to approximate $\hat{F}_{T^{*}}$ by an average of independent and identically distributed random variables with mean zero.
To obtain the approximation of $\hat{F}_{T^{*}}$, we need the following assumptions about the kernel and bandwidth.\\[0.2cm]
\noindent
\textbf{Assumption K} (Kernel)\\
The kernel function $K:\mathbb{R}^{d}\rightarrow \mathbb{R}$ satisfies the following conditions.
\begin{enumerate}[label=K\arabic*]
\item \label{K1} $K$ is of bounded variation, vanishes outside $[-1,1]^{d}$, and satisfies $\int K(u)\Myd u=1$.
\item \label{K2} There is a positive integer $r\geq 2$ such that
\begin{align*}
\int \left(\prod_{\ell=1}^{d}u_{\ell}^{\lambda_{\ell}}\right)K(u)\Myd u=0
\end{align*}
for any $d$-dimensional vector $\lambda=(\lambda_{1},\ldots,\lambda_{d})^{\top}$ of nonnegative integers with $\sum_{\ell=1}^{d}\lambda_{\ell}\leq r-1$.
\item \label{K3} For $u\in[-1,1]^{d}$, $K(u)$ is $r$-times differentiable with respect to $u$ and the derivatives are uniformly continuous and bounded.
\item \label{K4} For $u\in[-1,1]^{d}$, $K(u)=K(|u|)$.
\end{enumerate}
\textbf{Assumption B} (Bandwidth)\\
The sequence $\{h_{n}\}_{n=1}^{\infty}$ of bandwidths satisfies the following conditions.
\begin{enumerate}[label=B\arabic*]
\item \label{B1} $h_{n}\rightarrow 0$.
\item \label{B2} $n^{1/2}\left(\frac{\log n}{nh_{n}^{d}}\right)^{3/4}\rightarrow 0$.
\item \label{B3} $n^{1/2}h_{n}^{r}\rightarrow 0$.
\end{enumerate}
These assumptions about the kernel and bandwidth are mild.
Assumption~\ref{K3} restricts the choice of kernels so that the estimator $\hat{F}_{T|X}(t|x)$ is $r$-times differentiable with respect to $x$ and these derivatives are uniformly continuous and bounded.
Assumptions~\ref{B2} implies the remainder term in the representation of $\hat{F}_{T|X}$ in Proposition~\ref{PRO3} below is of order less than $n^{-1/2}$.
Assumptions~\ref{K2}-\ref{K3} and \ref{B3} are imposed to ensure the bias terms of $\hat{F}_{T^{*}}$ in Proposition~\ref{PRO3} are also of order less than $n^{-1/2}$.
Note that a necessary condition to make Assumptions~\ref{B2} and \ref{B3} valid simultaneously is $3d<2r$.
Thus, we use a higher order kernel to construct $\hat{F}_{T^{*}}$ under Assumption~\ref{K2} if multivariate policy variables are of interest, that is, $d\geq 2$.
Assumption~\ref{K4} is valid if $K$ is a product kernel function $K(u)=\prod_{\ell=1}^{d}k_{\ell}(u_{\ell})$ and each $k_{\ell}$ is a univariate kernel function that is symmetric around zero.

Moreover, we need conditions about the support and smoothness of densities and distributions as follows.\\[0.2cm]
\textbf{Assumption SP} (Support)
\begin{enumerate}[label=SP\arabic*]
\item \label{SP1} The support of $X^{*}$ is the compact subset $J^{*}\equiv\prod_{\ell=1}^{d}[\underline{X_{\ell}^{*}},\overline{X_{\ell}^{*}}]$ of the interior of the support of $X$, say $J\equiv\prod_{\ell=1}^{d}[\underline{X_{\ell}},\overline{X_{\ell}}]$.
\item \label{SP2}
There exist positive numbers $u_{0}$ and $v_{0}$ such that $\inf\left\{m(x):x\in J_{v_{0}}^{*}\right\}\geq u_{0}$ where $J_{v_{0}}^{*}=\prod_{\ell=1}^{d}[\underline{X_{\ell}^{*}}-v_{0},\overline{X_{\ell}^{*}}+v_{0}]$.
\item \label{SP3} There exist positive numbers $\zeta^{*}$ and $v^{*}$ such that $\inf\{1-F_{Y|X}(\zeta^{*}|x):x\in J\}\geq v^{*}$.
\end{enumerate}
Assumption~\ref{SP1}, stronger than Assumption~\ref{I4}, requires the support of $X^{*}$ to be a proper subset of the support of $X$.
When $J^{*}$ is close to $J$, Assumption~\ref{SP2} is valid provided the density of $X$ on the boundary of $J^{*}$ is still bounded away from zero.
Assumption~\ref{SP3} requires the conditional survival function of $Y$ given $X$
is uniformly bounded away from zero on $[0,\zeta^{*}]\times J$; in addition, it implies $\zeta^{*}<\tau=\inf\{t:F_{Y}(t)=1\}$.\\[0.2cm]
\textbf{Assumption SM} (Smoothness)
\begin{enumerate}[label=SM\arabic*]
\item \label{SM1} The function $m(x)$ is $r$-times differentiable with respect to $x$ on the interior of $J$, and its derivatives are bounded and uniformly continuous.
\item \label{SM2} For all $(t,x)\in \mathbb{R}\times J_{v_{0}}^{*}$, the first $r$ partial derivative with respect to $x$ of $F_{T|X}(t|x)$, $F_{C|X}(t|x)$, $f_{T|X}(t|x)$ and $f_{C|X}(t|x)$ are bounded.
\item \label{SM3} For all $(t,x)\in \mathbb{R}\times J_{v_{0}}^{*}$, the first derivative with respect to $t$ of $f_{T|X}(t|x)$ and $f_{C|X}(t|x)$ are bounded.
\item \label{SM4} The function $m^{*}(x)$ is $r$-times differentiable with respect to $x$ on the interior of $J$, and its derivatives are bounded and uniformly continuous.\footnote{Let $m^{*}(x)=0$ if $x\notin J^{*}$.}
\item \label{SM5} Both $\int (\sup_{s\in[0,\zeta^{*}]}f_{T|X}(s|x))^{2}m(x)\Myd x$ and $\int [m^{*}(x)]^{2}/m(x)\Myd x$ are finite.
\end{enumerate}
Assumptions~\ref{SM1}-\ref{SM3} are imposed to obtain the representation of $\hat{F}_{T|X}$ in Proposition~\ref{PRO3}.
Similar conditions are used in \citet{Iglesias-PerezGonzalez-Manteiga1999} and \citet{LiangUna-AlvarezEtAl2012}.
As in \citet{Rothe2010}, we impose Assumption~\ref{SM4} to establish the representation of $\hat{F}_{T^{*}}$ in Proposition~\ref{PRO3} by standard kernel smoothing techniques.
Assumption~\ref{SM5} is technical and valid if the second moments of $\sup_{s\in[0,\zeta^{*}]}f_{T|X}(s|X)$ and $m^{*}(X)/m(X)$ are finite.


\begin{Pro}\label{PRO3}
\text{(i)}
Under Assumptions~\ref{D1}-\ref{D2}, \ref{I1}-\ref{I2}, \ref{K1}-\ref{K3}, \ref{B1}-\ref{B3}, \ref{SP1}-\ref{SP3}, and \ref{SM1}-\ref{SM3}, for $(t,x)\in [0,\zeta^{*}]\times J^{*}$,
\begin{displaymath}
\hat{F}_{T|X}(t|x)-F_{T|X}(t|x)=\sum_{i=1}^{n}\xi^{*}(Y_{i},\delta_{i};t,x)B_{n_{i}}(x)+r_{n}(t,x),
\end{displaymath}
where
\begin{displaymath}
\xi^{*}(y,\delta;t,x)=\left[1-F_{T|X}(t|x)\right]\left[\frac{\Ind{[y\leq t, \delta=1]}}{1-F_{Y|X}(y|x)}-\int_{0}^{y\wedge t}\frac{F^{\delta}_{Y|X}(\!\!\Myd u|x)}{(1-F_{Y|X}(u|x))^{2}}\right],
\end{displaymath}
and
\begin{displaymath}
\sup_{(t,x)\in [0,\zeta^{*}]\times J^{*}}|r_{n}(t,x)|=\Bigoas{\left(\frac{\log n}{nh^{d}}\right)^{3/4}}.
\end{displaymath}
\text{(ii)}
If in addition Assumptions~\ref{I3}-\ref{I4}, \ref{K4}, and~\ref{SM4}-\ref{SM5} hold, then for $t\in [0,\zeta^{*}]$,
\begin{align*}
& \sqrt{n}\left(\hat{F}_{T^{*}}(t)-F_{T^{*}}(t)\right)\\
= & \frac{1}{\sqrt{n}}\sum_{i=1}^{n}\left(F_{T|X}(t|X_{i}^{*})-F_{T^{*}}(t)\right)
+\frac{1}{\sqrt{n}}\sum_{i=1}^{n}\xi^{*}(Y_{i},\delta_{i};t,X_{i})\frac{m^{*}(X_{i})}{m(X_{i})}+R^{*}_{n}(t),
\end{align*}
where
\begin{displaymath}
\xi^{*}(y,\delta;t,x)=\left[1-F_{T|X}(t|x)\right]\left[\frac{\Ind{[y\leq t, \delta=1]}}{1-F_{Y|X}(y|x)}-\int_{0}^{y\wedge t}\frac{F^{\delta}_{Y|X}(\!\!\Myd u|x)}{(1-F_{Y|X}(u|x))^{2}}\right].
\end{displaymath}
and
\begin{displaymath}
\sup_{(t,x)\in [0,\zeta^{*}]}|R^{*}_{n}(t)|=\Smallop{1}.
\end{displaymath}
\qed
\end{Pro}

\noindent
The influence function of $\hat{F}_{T^{*}}$ in Proposition~\ref{PRO3} has mean zero and can be decomposed into two components:
the first part arises from the sample variation in $X^{*}$, and the second part results from the estimate of $F_{T|X}(t|x)$.
The influence function of $\hat{F}_{T^{*}}$ is different from that of the counterfactual CDF estimator in \citet{Rothe2010} because we use the estimator $\hat{F}_{T|X}(t|x)$ in (\ref{CKM}) to recover the conditional CDF $F_{T|X}(t|x)$ in the presence of random censoring.

\subsection{Asymptotic properties}\label{Asymptotic}
Propositions~\ref{PRO2} and \ref{PRO3} show that the estimators $\hat{F}_{T}(t)$ and $\hat{F}_{T^{*}}(t)$ can be represented by the average of functions of independent and identically distributed random variables $(Y, \delta, X, X^{*})^{\top}$ plus asymptotic negligible terms.
The counterfactual estimator $\hat{F}_{T^{*}}$ is uniformly consistent for $F_{T^{*}}$ on $[0,\zeta^{*}]$ because the two classes $\{x\mapsto F_{T|X}(t|x):t\in[0,\zeta^{*}]\}$ and $\{(y,\delta,x)\mapsto\xi^{*}(y,\delta;t,x):t\in [0,\zeta^{*}]\}$ are both Euclidean under the assumptions imposed.
Moreover, for each $t\in[0,\zeta^{*}]$, the proposed estimator $\hat{F}_{T^{*}}(t)$ can avoid the curse of dimensionality, namely convergence at the usual parametric rate $\sqrt{n}$, even if the first-stage estimator $\hat{F}_{T|X}(t|x)$ converges to $F_{T|X}(t|x)$ at a rate slower than $\sqrt{n}$  for each $x\in J^{*}$.\footnote{
Details can be found in \citet{Dabrowska1989}, \citet{Iglesias-PerezGonzalez-Manteiga1999}, \citet{Iglesias-Perez2003}, and \citet{LiangUna-AlvarezEtAl2012}.
}

The representations of $\hat{F}_{T}$ and $\hat{F}_{T^{*}}$ further allow us to apply techniques in the literature on empirical processes to show that the random map
\begin{equation}\label{process}
t\mapsto \sqrt{n}\left(\hat{\mathbf{F}}(t)-\mathbf{F}(t)\right)
\end{equation}
converges weakly to a two dimensional centered Gaussian process, where $\hat{\mathbf{F}}\equiv(\hat{F}_{T^{*}},\hat{F}_{T})^{\top}$, $\mathbf{F}\equiv(F_{T^{*}},F_{T})^{\top}$, and $t=(t_{1},t_{2})^{\top}$.
Let $Z=(Y,\delta,X,X^{*})^{\top}$ and
\begin{align}\label{inf}
\Psi(t,Z)
=\begin{bmatrix}
\psi_{1}(t_{1},Z)\\
\psi_{2}(t_{2},Z)
\end{bmatrix}
=\begin{bmatrix}
\left(F_{T|X}(t_{1}|X^{*})-F_{T^{*}}(t_{1})\right)+\xi^{*}(Y,\delta;t_{1},X)\frac{m^{*}(X)}{m(X)}\\
\xi(Y,\delta;t_{2})
\end{bmatrix}
\end{align}
where $\xi$ and $\xi^{*}$ are defined in Propositions~\ref{PRO2} and~\ref{PRO3}, respectively.
We establish the weak convergence of $\hat{\mathbf{F}}$ as follows.

\begin{Thm}\label{THM1}
If Assumptions D, I, K, B, SP, and SM hold, then in $\mathbb{D}[0,\zeta^{*}]\times\mathbb{D}[0,\zeta^{*}]$,
\begin{align*}
\sqrt{n}\left(\hat{\mathbf{F}}(\cdot)-\mathbf{F}(\cdot)\right)\Rightarrow \mathbb{F}(\cdot),
\end{align*}
where $\mathbb{F}$ is a two dimensional centered Gaussian process with covariance function
$\Sigma(s,t)=\mathbb{E}\left(\Psi(s,Z)\Psi(t,Z)^{\top}\right)$ for every $s,t\in[0,\zeta^{*}]\times[0,\zeta^{*}]$.\qed
\end{Thm}

Theorem~\ref{THM1} demonstrates that when the sample size is large, we can approximate the random map in~(\ref{process}) by the two dimensional centered Gaussian process $\mathbb{F}$ with the covariance function
\begin{align*}
\Sigma(s,t)=\mathbb{E}\left(\Psi(s,Z)\Psi(t,Z)^{\top}\right)
=\begin{bmatrix}
\Sigma_{11}(s_{1},t_{1})&\Sigma_{12}(s_{1},t_{2})\\
\Sigma_{21}(s_{2},t_{1})&\Sigma_{22}(s_{2},t_{2})
\end{bmatrix}
\end{align*}
where
\begin{align*}
&\Sigma_{11}(u,u')\\
=&\Exp\left[\left(F_{T|X}(u|X^{*})-F_{T^{*}}(u)\right)\left(F_{T|X}(u'|X^{*})-F_{T^{*}}(u')\right)\right]\\
&\hspace{0.3cm}+\Exp\left[\left(\frac{m^{*}(X)}{m(X)}\right)^{2}		
\left[1-F_{T|X}(u|X)\right]\left[1-F_{T|X}(u'|X)\right]
\int_{0}^{u\wedge u'}\frac{F^{\delta}_{Y|X}(\!\!\Myd \tilde{u}|X)}{(1-F_{Y|X}(\tilde{u}|X))^{2}}\right],\\[0.3cm]
&\Sigma_{22}(u,u')\\
=&\left[1-F_{T}(u)\right]\left[1-F_{T}(u')\right]\int_{0}^{u\wedge u'}\frac{F^{\delta}_{Y}(\!\!\Myd \tilde{u})}{(1-F_{Y}(\tilde{u}))^{2}},
\end{align*}
and
\begin{align*}
&\Sigma_{12}(u,u')=\Sigma_{21}(u',u)\\
=&\Exp\Bigg[\left(
\left[1-F_{T|X}(u|X)\right]\left[\frac{\Ind{[Y\leq u, \delta=1]}}{1-F_{Y|X}(Y|X)}-\int_{0}^{Y\wedge u}\frac{ F^{\delta}_{Y|X}(\!\!\Myd \tilde{u}|X)}{(1-F_{Y|X}(\tilde{u}|X))^{2}}\right]
\right)\\
&\hspace{3.2cm}\cdot\left(\frac{m^{*}(X)}{m(X)}  \right)
\left(\frac{\Ind{[Y\leq u',\delta=1]}}{1-F_{Y}(Y)}-\int_{0}^{Y\wedge u'}\frac{ F^{\delta}_{Y}(\!\!\Myd \tilde{u})}{(1-F_{Y}(\tilde{u}))^{2}}\right)
\Bigg]\left[1-F_{T}(u')\right]\\
&+\Exp\left[\left(F_{T|X}(u|X^{*})-F_{T^{*}}(u)\right)
\left(\frac{\Ind{[Y\leq u', \delta=1]}}{1-F_{Y}(Y)}-\int_{0}^{Y\wedge u'}\frac{ F^{\delta}_{Y}(\!\!\Myd \tilde{u})}{(1-F_{Y}(\tilde{u}))^{2}}\right)\right]\left[1-F_{T}(u')\right].
\end{align*}
The second term in the last line is zero if $X^{*}$ is independent of $(Y,\delta,X)$, which is expected to be valid in the analysis of the composition effect.
In contrast, if we consider a counterfactual manipulation with $X^{*}=\pi(X)$, then the second term should not be omitted in general.

We can make inference on the distribution policy effect $\triangle_{F}(t)$ by $\hat{\triangle}_{F}(t)\equiv \hat{F}_{T^{*}}(t)-\hat{F}_{T}(t)$ because $\sqrt{n}\left(\hat{\triangle}_{F}(\cdot)-\triangle_{F}(\cdot)\right)$ converges weakly to $(1,-1)\mathbb{F}(\cdot)$ in $\mathbb{D}[0,\zeta^{*}]\times\mathbb{D}[0,\zeta^{*}]$ by
Theorem~\ref{THM1} and the continuous mapping theorem.
The covariance function $\Sigma(s,t)$ can be estimated by replacing the unknown functions with associated consistent estimators.
For example, a plug-in estimator of $\Sigma_{11}(u,u')$ is
\begin{align*}
&\hat{\Sigma}_{11}(u,u')\\
=&\frac{1}{n}\sum_{i=1}^{n}\left[\hat{F}_{T|X}(u|X_{i}^{*})-\hat{F}_{T^{*}}(u)\right]
\left[\hat{F}_{T|X}(u'|X_{i}^{*})-\hat{F}_{T^{*}}(u')\right]\\
&+\frac{1}{n}\sum_{i=1}^{n}
\left(\frac{\hat{m}^{*}(X_{i})}{\hat{m}(X_{i})}\right)^{2}
\left[1-\hat{F}_{T|X}(u|X_{i})\right]\left[1-\hat{F}_{T|X}(u'|X_{i})\right]
\int_{0}^{u\wedge u'}\frac{\hat{F}^{\delta}_{Y|X}(\!\!\Myd \tilde{u}|X_{i})}{(1-\hat{F}_{Y|X}(\tilde{u}|X_{i}))^{2}}
\end{align*}
where $\hat{F}_{T^{*}}$ is defined in~(\ref{CounterKM}), $\hat{F}_{T|X}$ is defined in~(\ref{CKM}),
\begin{align*}
\hat{F}_{Y|X}(y|x)&\equiv\frac{\sum_{i=1}^{n}\Ind{[Y_{i}\leq y]}K\left(\frac{x-X_{i}}{h_{n}}\right)}
{\sum_{i=1}^{n}K\left(\frac{x-X_{i}}{h_{n}}\right)},\\
\hat{F}^{\delta}_{Y|X}(y|x)&\equiv\frac{\sum_{i=1}^{n}\Ind{[Y_{i}\leq y,\delta_{i}=1]}K\left(\frac{x-X_{i}}{h_{n}}\right)}
{\sum_{i=1}^{n}K\left(\frac{x-X_{i}}{h_{n}}\right)},\\
\hat{m}(x)&\equiv\frac{1}{nh_{n}^{d}}\sum_{i=1}^{n}K\left(\frac{x-X_{i}}{h_{n}}\right),\;\;\text{and}\\
\hat{m}^{*}(x)&\equiv\frac{1}{nh_{n}^{d}}\sum_{i=1}^{n}K\left(\frac{x-X^{*}_{i}}{h_{n}}\right).
\end{align*}

To analyze other counterfactual policy effects, we apply the functional delta method as follows.
\begin{Thm}\label{THM2}
Let $\nu$ be a functional mapping from a subset of $\mathbb{D}[0,\zeta^{*}]\times\mathbb{D}[0,\zeta^{*}]$ to some normed space $\mathcal{V}$.
Suppose that $\nu$ is Hadamard differentiable at $\mathbf{F}$ with derivative $\nu_{\mathbf{F}}'$.
Let $Z=(Y,\delta,X,X^{*})^{\top}$ and $\Psi^{\nu}(t,Z)=\nu_{\mathbf{F}}'(\Psi)(t,Z)$, where $\Psi=(\psi_{1},\psi_{2})^{\top}$ is defined in~(\ref{inf}).
Under the assumptions of Theorem~\ref{THM1}, we have that in $\mathcal{V}$,
\begin{displaymath}
\sqrt{n}\left(\nu(\hat{\mathbf{F}})(\cdot)-\nu(\mathbf{F})(\cdot)\right) \Rightarrow \nu_{\mathbf{F}}'(\mathbb{F})(\cdot)\equiv\mathbb{G}(\cdot),
\end{displaymath}
where $\mathbb{G}$ is a two dimensional centered Gaussian process with covariance function $\Sigma^{\nu}(s,t)=\mathbb{E}\left(\Psi^{\nu}(s,Z)\Psi^{\nu}(t,Z)^{\top}\right)$.\qed
\end{Thm}

\noindent
Let $\hat{\Lambda}_{T}= \nu(\hat{F}_{T})$ and $\hat{\Lambda}_{T^{*}}= \nu(\hat{F}_{T^{*}})$ for the functional $\nu:F\mapsto\int_{[0,\cdot]} \frac{1}{1-F^{-}}\Myd F$.
In addition, let $\hat{\mathbf{\Lambda}}\equiv(\hat{\Lambda}_{T^{*}},\hat{\Lambda}_{T})^{\top}$ and $\mathbf{\Lambda}\equiv(\Lambda_{T^{*}},\Lambda_{T})^{\top}$.
Theorem~\ref{THM2} immediately implies the following corollary.
We can make inference on the cumulative hazard policy effect $\triangle_{\Lambda}(t)$ by $\hat{\triangle}_{\Lambda}(t)\equiv\hat{\Lambda}_{T^{*}}(t)-\hat{\Lambda}_{T}(t)$ because $\sqrt{n}\left(\hat{\triangle}_{\Lambda}(\cdot)-\triangle_{\Lambda}(\cdot)\right)$ converges weakly in $\mathbb{D}[0,\zeta^{*}]\times\mathbb{D}[0,\zeta^{*}]$ by the continuous mapping theorem.

\begin{Cor}\label{COR1}
Under the assumptions of Theorem~\ref{THM1},
\begin{align*}
\sqrt{n}\left(\hat{\mathbf{\Lambda}}(\cdot)-\mathbf{\Lambda}(\cdot)\right)
\Rightarrow &
\left[
\begin{array}{c}
\int_{0}^{\cdot}\frac{1}{1-F^{-}_{T^{*}}(u)}\mathbb{F}_{1}(\!\!\Myd u)
+\int_{0}^{\cdot}\frac{\mathbb{F}^{-}_{1}(u)}{\left(1-F^{-}_{T^{*}}(u)\right)^{2}}F_{T^{*}}(\!\!\Myd u) \\
\int_{0}^{\cdot}\frac{1}{1-F^{-}_{T}(u)}\mathbb{F}_{2}(\!\!\Myd u)
+\int_{0}^{\cdot}\frac{\mathbb{F}^{-}_{2}(u)}{\left(1-F^{-}_{T}(u)\right)^{2}}F_{T}(\!\!\Myd u)
\end{array} \right]
\equiv\mathbb{A}(\cdot)
\end{align*}
in $\mathbb{D}[0,\zeta^{*}]\times\mathbb{D}[0,\zeta^{*}]$.
The two dimensional process $\mathbb{A}$ is centered Gaussian with covariance function $\Sigma^{\Lambda}(s,t)=\mathbb{E}\left(\Psi^{\Lambda}(s,Z)\Psi^{\Lambda}(t,Z)^{\top}\right)$, where $\Psi^{\Lambda}(t,Z)=\left(\frac{\psi_{1}(t_{1},Z)}{1-F_{T^{*}}(t_{1})}, \frac{\psi_{2}(t_{2},Z)}{1-F_{T}(t_{2})}\right)^{\top}$, $Z=(Y,\delta,X,X^{*})^{\top}$, and $(\psi_{1},\psi_{2})^{\top}$ is defined in~(\ref{inf}).\qed
\end{Cor}

\section{Monte Carlo simulation}\label{Simulation}
In this section, we evaluate the small-sample performance of the proposed estimator $\hat{F}_{T^{*}}$ and its associated counterfactual policy effects by Monte Carlo simulation.
We consider the following data generating process (DGP):
\begin{align*}
Y= & \min{\{T,C\}},\\
T= & 5-3X_{1}+2X_{2}+\varepsilon\cdot\sqrt{X_{1}^{2}+X_{2}^{2}},
\end{align*}
where the covariates $(X_{1}$, $X_{2})$ follow the Beta distribution with shape parameters $(2,2)$,
the unobserved heterogeneity $\varepsilon$ is exponentially distributed with mean 2, and $C$ is log-normally distributed with parameters $(2.5,1)$; additionally, $(X_{1}, X_{2}, \varepsilon, C)$ are mutually independent.
The censoring rate in this design is approximate $23.45\%$.
We study the policy intervention
\begin{align*}
(X^{*}_{1},X^{*}_{2})^{\top}=\pi(X_{1},X_{2})=0.05+0.9\cdot(X_{1},X_{2})^{\top}
\end{align*}
and this intervention does not affect $(C,\varepsilon)$.
The DGP and counterfactual policy are not meant to mimic any data set in empirical studies; instead, they are only used to illustrate the proposed method.

We consider the sample sizes $n=100, 200, 400$, and $800$.
The number of simulation replications is $S=1000$.
The criteria of evaluation include the mean integrated absolute error (MIAE) and the root mean integrated squared error (RMISE).\footnote{
For an estimator $\hat{f}$ of a generic real-valued function $f$, the mean integrated absolute error of $\hat{f}$ is
\begin{align*}
\text{MIAE}(\hat{f})=\Exp\left[\int |\hat{f}(u)-f(u)| \Myd u\right]
\end{align*}
and the root mean integrated squared error of $\hat{f}$ is
\begin{align*}
\text{RMISE}(\hat{f})=\sqrt{\Exp\left[\int |\hat{f}(u)-f(u)|^{2} \Myd u\right]}.
\end{align*}
}
The CDF and cumulative hazard estimators in this Monte Carlo study are calculated over the eqidistant grids $\{4.25,4.30,4.35,\ldots,8.10,8.15\}$, and the numerical integration in MIAE and RMISE is taken over $[4.25,8.15]$, where $4.25$ and $8.15$ are the $10\%$ and $90\%$ quantile of $T$, respectively.

We evaluate the estimation of the CDFs $(F_{T^{*}},F_{T})$ and the estimation of the cumulative hazard functions $(\Lambda_{T^{*}},\Lambda_{T})$.
We estimate $F_{T}$ by the unconditional Kaplan-Meier estimator $\hat{F}_{T}$ in~(\ref{KM}).
To estimate $F_{T^{*}}$, we use the proposed estimator $\hat{F}_{T^{*}}$ in~(\ref{CounterKM}) with the fourth order product kernel function $K(u_{1},u_{2})=k(u_{1})k(u_{2})$ where $k(u)=(15/32)(3-10u^{2}+7u^{4})\Ind{[|u|<1]}$ and the bandwidth $h_{n}=3n^{-1/7}$.\footnote{
Assumptions~\ref{K1}-\ref{K3} and \ref{B1}-\ref{B3} are satisfied under this choice of kernel function and bandwidth}
Table~\ref{CDF} shows that the MIAE and RMISE of $(\hat{F}_{T^{*}},\hat{F}_{T})$ shrink as the sample size increases.
\begin{table}[t]
\centering
\caption{Estimation of the CDFs}
\label{CDF}
\begin{tabular}{ccccc}
\toprule
MIAE & Unconditional CDF  & \multicolumn{3}{l}{\hspace{0.5cm} Counterfactual CDF}\\[0.2cm]
     & Kaplan-Meier (1958)    & Proposed & Oracle & Rothe (2010)\\[0.2cm]
$n$  & $\hat{F}_{T}$ & $\hat{F}_{T^{*}}$ & $\tilde{F}_{T^{*}}$ & $F^{\dagger}_{T^{*}}$\\[0.3cm]
100  & 0.1454 & 0.1482 & 0.1440 & 0.3414\\
200  & 0.1061 & 0.1089 & 0.1052 & 0.3385\\
400  & 0.0741 & 0.0761 & 0.0734 & 0.3392\\
800  & 0.0528 & 0.0547 & 0.0523 & 0.3409\\[0.3cm]
\midrule
RMISE & Unconditional CDF  & \multicolumn{3}{l}{\hspace{0.5cm} Counterfactual CDF}\\[0.2cm]
     & Kaplan-Meier (1958)    & Proposed & Oracle & Rothe (2010)\\[0.2cm]
$n$  & $\hat{F}_{T}$ & $\hat{F}_{T^{*}}$ & $\tilde{F}_{T^{*}}$ & $F^{\dagger}_{T^{*}}$\\[0.3cm]
100  & 0.0922 & 0.0942 & 0.0914 & 0.2019\\
200  & 0.0674 & 0.0693 & 0.0668 & 0.1950\\
400  & 0.0475 & 0.0489 & 0.0471 & 0.1903\\
800  & 0.0336 & 0.0349 & 0.0333 & 0.1888
\end{tabular}
\end{table}
The MIAE and RMISE of $\hat{F}_{T^{*}}$ halves as the sample size quadruples; namely, this estimator converges at the rate $\sqrt{n}$.
This confirms the theoretical analysis that the proposed estimator $\hat{F}_{T^{*}}$ does not suffer from the curse of dimensionality.
We also consider an oracle estimator $\tilde{F}_{T^{*}}$, which is the unconditional Kaplan-Meier estimator of $F_{T^{*}}$ if $Y^{*}=\min\{T^{*},C\}$ and $\delta^{*}=\Ind{[T^{*}\leq C]}$ are observed.
Surprisingly, this oracle estimator $\tilde{F}_{T^{*}}$ does not outweigh considerably the proposed estimator $\hat{F}_{T^{*}}$ in terms of MIAE and RMISE; however, $\tilde{F}_{T^{*}}$ is infeasible because $Y^{*}$ and $\delta^{*}$ are unobserved in practice.
Moreover, we consider Rothe's (\citeyear{Rothe2010}) counterfactual estimator $F^{\dagger}_{T^{*}}$, which is constructed under the assumption that data are not censored.\footnote{
Since the support of $(X^{*}_{1},X^{*}_{2})$ is a proper subset of the support of $(X_{1},X_{2})$, we construct Rothe's estimator based on the aforementioned fourth order kernel function and bandwidth $h_{n}=3n^{-1/7}$.
}
As shown in Table~\ref{CDF}, the neglect of censoring results in larger MIAE and RMISE of $F^{\dagger}_{T^{*}}$, compared with those of $\hat{F}_{T^{*}}$ and $\tilde{F}_{T^{*}}$.
Table~\ref{Hazard} reports the MIAE and RMISE of the estimated cumulative hazard functions
\begin{align*}
\hat{\Lambda}_{T}\equiv \nu(\hat{F}_{T}),\;\;
\hat{\Lambda}_{T^{*}}\equiv \nu(\hat{F}_{T^{*}}),\;\;
\tilde{\Lambda}_{T^{*}}\equiv \nu(\tilde{F}_{T^{*}}),\;\;\text{and}\;\;
\Lambda^{\dagger}_{T^{*}}\equiv \nu(F^{\dagger}_{T^{*}})
\end{align*}
where $\nu$ is the functional that maps $F$ to $-\log{(1-F)}$.
\begin{table}[t]
\centering
\caption{Estimation of the Cumulative Hazard Functions}
\label{Hazard}
\begin{tabular}{ccccc}
\toprule
MIAE & Unconditional CDF  & \multicolumn{3}{l}{\hspace{0.5cm} Counterfactual CDF}\\[0.2cm]
     & Kaplan-Meier (1958)   & Proposed & Oracle & Rothe (2010)\\[0.2cm]
$n$  & $\hat{\Lambda}_{T}$ & $\hat{\Lambda}_{T^{*}}$ & $\tilde{\Lambda}_{T^{*}}$ & $\Lambda^{\dagger}_{T^{*}}$\\[0.3cm]
100  & 0.5382 & 0.5359 & 0.5491 & 1.2047\\
200  & 0.3870 & 0.3945 & 0.3949 & 1.1419\\
400  & 0.2661 & 0.2736 & 0.2712 & 1.1236\\
800  & 0.1862 & 0.1944 & 0.1900 & 1.1227\\[0.3cm]
\midrule
RMISE & Unconditional CDF  & \multicolumn{3}{l}{\hspace{0.5cm} Counterfactual CDF}\\[0.2cm]
     & Kaplan-Meier (1958)    & Proposed & Oracle & Rothe (2010)\\[0.2cm]
$n$  & $\hat{\Lambda}_{T}$ & $\hat{\Lambda}_{T^{*}}$ & $\tilde{\Lambda}_{T^{*}}$ & $\Lambda^{\dagger}_{T^{*}}$\\[0.3cm]
100  & 0.4124 & 0.3919 & 0.4236 & 0.7622\\
200  & 0.2861 & 0.2845 & 0.2938 & 0.6700\\
400  & 0.1943 & 0.1974 & 0.1990 & 0.6212\\
800  & 0.1332 & 0.1377 & 0.1366 & 0.6025
\end{tabular}
\end{table}
Similarly, the simulation results provide evidence that the proposed cumulative hazard function $\hat{\Lambda}_{T^{*}}$ converges at the rate $\sqrt{n}$. Moreover, $\hat{\Lambda}_{T^{*}}$ performs as well as the oracle estimator $\tilde{\Lambda}_{T^{*}}$.
The neglect of censoring, however, causes relatively large bias of $\Lambda^{\dagger}_{T^{*}}$.

\section{Conclusion}\label{Conclusion}
We have proposed a two-stage fully nonparametric estimator of the counterfactual CDF for a duration variable, which is subject to the random censoring.
Since the nonseparable heterogeneity is of unrestricted dimensionality and its marginal distribution is unspecified, the duration analysis in this paper would avoid several types of model misspecification in empirical studies.
The incorporation of covariates also enables researchers to evaluate the change of duration dependence in response to a counterfactual policy that changes exogenous covariates.

There are some directions of extension to this research.
First, we may relax the assumption of exogenous covariates by the control function approach.
Second, it would be important to establish the validity of a bootstrap method to construct a confidence band for the counterfactual policy effect.
Finally, the inclusion of time-varying covariates might be relevant in some empirical studies.

\begin{appendices}
\numberwithin{equation}{section}

\section{Technical Proofs}\label{AppendixA}
\subsection{Proof of Proposition~\ref{PRO1}}
\begin{proof}
Under Assumption~\ref{I1}, we can show that
\begin{align*}
F_{T}(t)=\int \Exp[\delta|Y=s]\exp\left\{\int_{0}^{s}\frac{1-\Exp[\delta|Y=y]}{1-F_{Y}(y)}F_{Y}(\!\!\Myd y)\right\}F_{Y}(\!\!\Myd s)
\end{align*}
for $t<\tau$.
See page 1604 of \citet{StuteWang1993}.
Since data on $\{(Y_{i},\delta_{i})\}_{i=1}^{n}$ are available, $F_{T}(t)$ is identified for $t\in[0,\tau)$.
Under Assumption~\ref{I2}, for $(t,x)\in [0,\tau)\times \mathbb{R}^{d}$, the identification of $F_{T|X}(t|x)$ is established by similar arguments in Lemma 25.74 of \citet{Vaart1998}.

The independence between $\varepsilon$ and $X$ implies
\begin{align*}
F_{T^{*}}(t)=\int \mathbb{P}(\varphi(x,\varepsilon)\leq t)F_{X^{*}}(\!\!\Myd x).
\end{align*}
Since $\varepsilon$ is also independent of $X^{*}$, which only takes values in a subset of support of $X$, we have
\begin{align*}
F_{T^{*}}(t)
=\int \mathbb{P}(\varphi(x,\varepsilon)\leq t|X=x)F_{X^{*}}(\!\!\Myd x)
=\Exp[F_{T|X}(t|X^{*})].
\end{align*}
\end{proof}

\subsection{Proof of Proposition~\ref{PRO2}}
\begin{proof}
See Theorem 1 of \citet{LoSingh1986} or Theorem 3 of \citet{Cai1998}.
\end{proof}

\subsection{Proof of Proposition~\ref{PRO3}}
\begin{proof}
(\textit{i})
Let $\hat{\Lambda}_{T|X}(t|x)=-\log{[1-\hat{F}_{T|X}(t|x)]}$.
By Theorem 2.1 of \citet{LiangUna-AlvarezEtAl2012},
\begin{align*}
\sup_{(t,x)\in[0,\zeta^{*}]\times J^{*}} |\hat{\Lambda}_{T|X}(t|x)-\Lambda_{T|X}(t|x)|=\Smalloas{1},
\end{align*}
where $\Lambda_{T|X}(t|x)=-\log{[1-F_{T|X}(t|x)]}$.
Hence, with probability one, $\hat{\Lambda}_{T|X}(t|x)$ is well defined on $[0,\zeta^{*}]\times J^{*}$ for $n$ large.
Taylor series expansion yields
\begin{align*}
&\hat{F}_{T|X}(t|x)-F_{T|X}(t|x)=[1-F_{T|X}(t|x)]\left[\hat{\Lambda}_{T|X}(t|x)-\Lambda_{T|X}(t|x)\right]\\
&\hspace{5cm}+\Bigo{\sup_{(t,x)\in [0,\zeta^{*}]\times J^{*}}\left[\hat{\Lambda}_{T|X}(t|x)-\Lambda_{T|X}(t|x)\right]^{2}}.
\end{align*}
The desired result follows from Theorems 2.1 and 2.3 of \citet{LiangUna-AlvarezEtAl2012}.\\[0.3cm]
(\textit{ii})
Let $\mathscr{D}_{n}\equiv \{Y_{i},\delta_{i},X_{i}\}_{i=1}^{n}$ and $X^{*}$ be a random variable that is independent of $\mathscr{D}_{n}$ and follows the same distribution as $X_{1}^{*}$ does.
Note that
\begin{align*}
&\sqrt{n}\left[\hat{F}_{T^{*}}(t)-F_{T^{*}}(t)\right]\\
=&\frac{1}{\sqrt{n}}\sum_{i=1}^{n}\left[\hat{F}_{T|X}(t|X_{i}^{*})-F_{T^{*}}(t)\right]\\
=&\frac{1}{\sqrt{n}}\sum_{i=1}^{n}\left[F_{T|X}(t|X_{i}^{*})-F_{T^{*}}(t)]\right]\\
&+\sqrt{n}\left\{\Exp[\hat{F}_{T|X}(t|X^{*})|\mathscr{D}_{n}]-\Exp[F_{T|X}(t|X^{*})]\right\}\\
&+\frac{1}{\sqrt{n}}\sum_{i=1}^{n}
\Big\{\left[\hat{F}_{T|X}(t|X_{i}^{*})-\Exp[\hat{F}_{T|X}(t|X^{*})|\mathscr{D}_{n}]\right]
-\left[F_{T|X}(t|X_{i}^{*})-\Exp[F_{T|X}(t|X_{i}^{*})]\right]\Big\}.
\end{align*}
Lemma~\ref{lemmaA1} shows that
\begin{align*}
&\frac{1}{\sqrt{n}}\sum_{i=1}^{n}
\Big\{\left[\hat{F}_{T|X}(t|X_{i}^{*})-\Exp[\hat{F}_{T|X}(t|X^{*})|\mathscr{D}_{n}]\right]
-\left[F_{T|X}(t|X_{i}^{*})-\Exp[F_{T|X}(t|X_{i}^{*})]\right]\Big\}\\
=&\Smallop{1}
\end{align*}
uniformly in $t\in[0,\zeta^{*}]$.
The result follows from Lemma~\ref{lemmaA2} that
\begin{align*}
\sqrt{n}\left\{\Exp[\hat{F}_{T|X}(t|X^{*})|\mathscr{D}_{n}]-\Exp[F_{T|X}(t|X^{*})]\right\}
=\frac{1}{\sqrt{n}}\sum_{i=1}^{n}\xi^{*}(Y_{i},\delta_{i};t,X_{i})\frac{m^{*}(X_{i})}{m(X_{i})}+r^{*}_{n}(t)
\end{align*}
and $\sup_{t\in[0,\zeta^{*}]}|r^{*}_{n}(t)|=\Smallop{1}$.
\end{proof}

\subsection{Proof of Theorem~\ref{THM1}}
\begin{proof}
The representations in Propositions~\ref{PRO2} and~\ref{PRO3} allow us to write $\sqrt{n}\left(\hat{\mathbf{F}}-\mathbf{F}\right)$ as a two-dimensional empirical process plus an asymptotically negligible term; that is, for each $t\in[0,\zeta^{*}]\times[0,\zeta^{*}]$, we have
\begin{align*}
\sqrt{n}\left(\hat{\mathbf{F}}(t)-\mathbf{F}(t)\right)
=\frac{1}{\sqrt{n}}\sum_{i=1}^{n}\left[\phi(t,Z_{i})-\Exp(\phi(t,Z_{i}))\right]
+\Smallop{1},
\end{align*}
where $Z_{i}=(Y_{i},\delta_{i},X_{i},X^{*}_{i})^{\top}$ for each $i\in\{1,2,\ldots,n\}$,
\begin{align*}
\phi(t,Z)=
\left[
\begin{array}{c}
F_{T|X}(t|X^{*})\\
0
\end{array} \right]
+
\left[
\begin{array}{c}
\xi^{*}(Y,\delta;t,X)\frac{m^{*}(X)}{m(X)} \\
\xi(Y,\delta;t)
\end{array} \right],
\end{align*}
and the two functions $\xi$ and $\xi^{*}$ are defined in Propositions~\ref{PRO2} and~\ref{PRO3}, respectively.
It follows from Lemma~\ref{Euclidean} that the classes $\{(y,\delta)\mapsto\xi(y,\delta;t):t\in [0,\zeta^{*}]\}$, $\{(y,\delta,x)\mapsto\xi^{*}(y,\delta;t,x):t\in [0,\zeta^{*}]\}$, and $\{x^{*}\mapsto F_{T|X}(t|x^{*}):t\in[0,\zeta^{*}]\}$ are all Euclidean.
Lemma 2.14 of \citet{PakesPollard1989} further implies that the class $\Phi\equiv \{\phi(t,\cdot): t\in[0,\zeta^{*}]\times[0,\zeta]\}$ is Euclidean under Assumption~\ref{SM5}.
Hence, the class $\Phi$ is Donsker by Theorem 19.14 of \citet{Vaart1998}, and the process
\begin{align*}
\sqrt{n}\left(\hat{\mathbf{F}}(\cdot)-\mathbf{F}(\cdot)\right)
=\frac{1}{\sqrt{n}}\sum_{i=1}^{n}\left[\phi(\cdot,Z_{i})-\Exp(\phi(\cdot,Z_{i}))\right]+\Smallop{1}
\end{align*}
converges weakly in $\mathbb{D}[0,\zeta^{*}]\times\mathbb{D}[0,\zeta^{*}]$ to a centered Gaussian process with covariance function
\begin{align*}
\Exp\left([\phi(s,Z)-\Exp(\phi(s,Z))][\phi(t,Z)-\Exp(\phi(t,Z))]^{\top}\right)
=\mathbb{E}\left(\Psi(s,Z)\Psi(t,Z)^{\top}\right)
\end{align*}
for each $s,t\in [0,\zeta^{*}]\times[0,\zeta^{*}]$.
\end{proof}

\subsection{Proof of Theorem~\ref{THM2}}
\begin{proof}
The proof is an application of the functional delta method, which is established in Theorem 20.8 of \citet{Vaart1998}.
\end{proof}

\subsection{Proof of Corollary~\ref{COR1}}
\begin{proof}
Let $\eta=\min\{1-F_{T^{*}}(\zeta^{*}), 1-F_{T}(\zeta^{*})\}$.
Assumptions~\ref{I1} implies that $1-F_{T}(\zeta^{*})\geq 1-F_{Y}(\zeta^{*})>0$; in addition, Assumptions~\ref{I2} and~\ref{SP3} imply that $1-F_{T^{*}}(\zeta^{*})=\mathbb{E}\left(1-F_{T|X}(\zeta^{*}|X^{*})\right)\geq v^{*} > 0$.
It follows that $\eta>0$.
Let $\mathbb{D}_{\eta}^{2}$ be the set of nondecreasing c\`{a}dl\`{a}g functions $(F_{1},F_{2})^{\top}$ such that
$F_{\ell}:[0,\zeta^{*}]\rightarrow \mathbb{R}$ with $F_{\ell}(0)=0$ and $1-F_{\ell}(\zeta^{*})\geq \eta$ for each $\ell\in\{1,2\}$.
Let $\nu$ be the functional from $\mathbb{D}_{\eta}^{2}$ to $\mathbb{D}[0,\zeta^{*}]\times\mathbb{D}[0,\zeta^{*}]$ such that $\nu(F_{1},F_{2})=(\nu_{0}(F_{1}),\nu_{0}(F_{2}))^{\top}$ where $\nu_{0}(F)(\cdot)=\int_{[0,\cdot]}\frac{F(\!\!\Myd u)}{1-F^{-}(u)}$.
From Lemma 20.14 of \citet{Vaart1998}, the functional $\nu$ is Hadamard differentiable at $\mathbf{F}=(F_{T^{*}},F_{T})^{\top}\in\mathbb{D}_{\eta}^{2}$.
Moreover, it can be shown that the Hadamard derivative is
\begin{align*}
\nu_{\mathbf{F}}'(S_{1},S_{2})
=
\left[
\begin{array}{c}
\int_{0}^{\cdot}\frac{1}{1-F^{-}_{T^{*}}(u)}S_{1}(\!\!\Myd u)
+\int_{0}^{\cdot}\frac{S^{-}_{1}(u)}{\left(1-F^{-}_{T^{*}}(u)\right)^{2}}F_{T^{*}}(\!\!\Myd u) \\
\int_{0}^{\cdot}\frac{1}{1-F^{-}_{T}(u)}S_{2}(\!\!\Myd u)
+\int_{0}^{\cdot}\frac{S^{-}_{2}(u)}{\left(1-F^{-}_{T}(u)\right)^{2}}F_{T}(\!\!\Myd u)
\end{array} \right]
\end{align*}
for $(S_{1},S_{2})^{\top}\in\mathbb{D}_{\eta}^{2}$.
Applying Theorem~\ref{THM2} yields that
\begin{align*}
\sqrt{n}\left(\hat{\mathbf{\Lambda}}(\cdot)-\mathbf{\Lambda}(\cdot)\right)
=\sqrt{n}\left(\nu(\hat{\mathbf{F}})(\cdot)-\nu(\mathbf{F})(\cdot)\right)
\Rightarrow
\nu_{\mathbf{F}}'(\mathbb{F})(\cdot)\equiv\mathbb{A}(\cdot)
\end{align*}
in $\mathbb{D}[0,\zeta^{*}]\times\mathbb{D}[0,\zeta^{*}]$.
\end{proof}

\subsection{Lemmas}
\begin{Lem}\label{lemmaA1}
Let $\mathscr{D}_{n}\equiv \{Y_{i},\delta_{i},X_{i}\}_{i=1}^{n}$ and $X^{*}$ be a random variable that is independent of $\mathscr{D}_{n}$ and follows the same distribution as $X_{1}^{*}$ does.
Under Assumptions~\ref{D1}-\ref{D2}, \ref{I1}-\ref{I2}, \ref{K1}-\ref{K3}, \ref{B1}-\ref{B3}, \ref{SP1}-\ref{SP3}, and \ref{SM1}-\ref{SM2},
\begin{align*}
&\sup_{t\in [0,\zeta^{*}]}\Bigg|\frac{1}{\sqrt{n}}\sum_{i=1}^{n}
\Big\{\left[\hat{F}_{T|X}(t|X_{i}^{*})-\Exp[\hat{F}_{T|X}(t|X^{*})|\mathscr{D}_{n}]\right]\\
&\hspace{5.6cm}-\left[F_{T|X}(t|X_{i}^{*})-\Exp[F_{T|X}(t|X^{*})]\right]\Big\}\Bigg|=\Smallop{1}.
\end{align*}
\end{Lem}
\begin{proof}
Let $\hat{\Gamma}(t|x;h_{n})=\hat{F}_{T|X}(t|x;h_{n})-F_{T|X}(t|x)$ for $(t,x)\in [0,\zeta^{*}]\times J^{*}$.
Our goal is to show that
\begin{align*}
&\sup_{t\in [0,\zeta^{*}]}\Bigg|\frac{1}{\sqrt{n}}\sum_{i=1}^{n}
\Big[\hat{\Gamma}(t|X_{i}^{*};h_{n})-\Exp[\hat{\Gamma}(t|X^{*};h_{n})|\mathscr{D}_{n}]\Big]\Bigg|=\Smallop{1}.
\end{align*}
For each $t\in [0,\zeta^{*}]$, we have
\begin{align*}
 \Exp\left[\left|\hat{\Gamma}(t|X^{*};h_{n})\right|^{2}|\mathscr{D}_{n}\right]
\leq & \sup_{(s,x)\in [0,\zeta^{*}]\times J^{*}}
\left|\hat{\Gamma}(s|x;h_{n})\right|^{2}
= \Bigoas{\frac{\log n}{nh^{d}}}
= \Smallop{1}
\end{align*}
by Theorem 2.1 of \citet{LiangUna-AlvarezEtAl2012}.
Let $\mathcal{C}_{r}(J^{*})$ be the class of real-valued functions defined on $J^{*}$ whose partial derivatives up to order $r$ exist and are bounded by some constant.
Example 19.9 of \citet{Vaart1998} shows that $\mathcal{C}_{r}(J^{*})$ is Donsker whenever $r>d/2$, which is guaranteed under Assumptions~\ref{B2} and~\ref{B3}.
Note that for each $t\in [0,\zeta^{*}]$, $\{x\mapsto\hat{\Gamma}(t|x;h_{n})\}_{n=1}^{\infty}$ is a sequence of random functions taking their values in $\mathcal{C}_{r}(J^{*})$ under Assumption~\ref{K3} and~\ref{SM2}.
It follows from Lemma 19.24 of \citet{Vaart1998} that for each $t\in [0,\zeta^{*}]$,
\begin{align*}
\frac{1}{\sqrt{n}}\sum_{i=1}^{n}
\left[\hat{\Gamma}(t|X_{i}^{*};h_{n})-\Exp[\hat{\Gamma}(t|X^{*};h_{n})|\mathscr{D}_{n}]\right]
=\Smallop{1}.
\end{align*}

It remains to show the uniform convergence in probability.
Since the dominated convergence theorem implies that with probability one,
\begin{align*}
\Exp\left[
\left|
\hat{\Gamma}(s|X^{*};h_{n})-\hat{\Gamma}(s'|X^{*};h_{n})
\right|^{2}
|\mathscr{D}_{n}
\right]\to 0
\end{align*}
as $s'\to s$,
the stochastic equicontinuity holds; specifically, we have for any $u>0$ there is a $v>0$ such that
\begin{align*}
&\limsup_{n\to\infty}\mathbb{P}\Bigg(
\sup_{|s-s'|<v}\Bigg|
\frac{1}{\sqrt{n}}\sum_{i=1}^{n}
\Big(\hat{\Gamma}(s|X_{i}^{*};h_{n})-\Exp[\hat{\Gamma}(s|X^{*};h_{n})|\mathscr{D}_{n}]\Big)\\
&\hspace{5cm}-\Big(\hat{\Gamma}(s'|X_{i}^{*};h_{n})-\Exp[\hat{\Gamma}(s'|X^{*};h_{n})|\mathscr{D}_{n}]\Big)
\Bigg|>u\Bigg)<u.
\end{align*}
Applying Theorem 21.9 of \citet{Davidson1994} yields the desired result.
\end{proof}

\begin{Lem}\label{lemmaA2}
Let $\mathscr{D}_{n}\equiv \{Y_{i},\delta_{i},X_{i}\}_{i=1}^{n}$ and $X^{*}$ be a random variable that is independent of $\mathscr{D}_{n}$ and follows the same distribution as $X_{1}^{*}$ does.
Suppose that the assumptions of Proposition~\ref{PRO3} hold.
Then, for $t\in [0,\zeta^{*}]$,
\begin{align*}
\sqrt{n}\left\{\Exp[\hat{F}_{T|X}(t|X^{*})|\mathscr{D}_{n}]-\Exp[F_{T|X}(t|X^{*})]\right\}
=\frac{1}{\sqrt{n}}\sum_{i=1}^{n}\xi^{*}(Y_{i},\delta_{i};t,X_{i})\frac{m^{*}(X_{i})}{m(X_{i})}+r_{n}^{*}(t)
\end{align*}
and $\sup_{t\in[0,\zeta^{*}]}|r_{n}^{*}(t)|=\Smallop{1}$.
\end{Lem}
\begin{proof}
Following the first part of Proposition~\ref{PRO3}, we obtain
\begin{align*}
 & \sqrt{n}\left\{\Exp[\hat{F}_{T|X}(t|X^{*})|\mathscr{D}_{n}]-\Exp[F_{T|X}(t|X^{*})]\right\}\\
= & \sqrt{n}\int_{J^{*}}\left[\hat{F}_{T|X}(t|x)-F_{T|X}(t|x)\right]m^{*}(x)\Myd x \\
= & \sqrt{n}\int_{J^{*}}\sum_{i=1}^{n}\xi^{*}(Y_{i},\delta_{i};t,X_{i})B_{n_{i}}(x)m^{*}(x)\Myd x \\
& +\sqrt{n}\int_{J^{*}}\sum_{i=1}^{n}
\left[\xi^{*}(Y_{i},\delta_{i};t,x)-\xi^{*}(Y_{i},\delta_{i};t,X_{i})\right]B_{n_{i}}(x)m^{*}(x)\Myd x\\
 & +\sqrt{n}\int_{J^{*}}r_{n}(t,x)m^{*}(x)\Myd x\\
= & \text{Term I}+\text{Term II}+\text{Term III}.
\end{align*}
Term III is asymptotically uniformly negligible because
\begin{align*}
\sup_{t\in[0,\zeta^{*}]}\left|\sqrt{n}\int_{J^{*}}r_{n}(t,x)m^{*}(x)\Myd x\right|
\leq& \sqrt{n}\left(\sup_{(t,x)\in[0,\zeta^{*}]\times J^{*}}\left|r_{n}(t,x)\right|\right)\int_{J^{*}}m^{*}(x)\Myd x\\
=&\Smallop{1}
\end{align*}
by Assumption~\ref{B2}.

We first show that $\text{Term I}$ can be approximated by
\begin{align*}
\frac{1}{\sqrt{n}}\sum_{i=1}^{n}\xi^{*}(Y_{i},\delta_{i};t,X_{i})\frac{m^{*}(X_{i})}{m(X_{i})}.
\end{align*}
For each $i\in\{1,2,\ldots,n\}$, we abbreviate by writing $\xi^{*}_{i,t}\equiv\xi^{*}(Y_{i},\delta_{i};t,X_{i})$, and we have $B_{n_{i}}(x)=K(\frac{x-X_{j}}{h_{n}})/nh_{n}^{d}\hat{m}(x)$.
Applying the second order Taylor expansion of $1/\hat{m}(x)$ around $1/m(x)$ yields
\begin{align*}
\text{Term I}= & \sqrt{n}\sum_{i=1}^{n}\xi^{*}_{i,t}\int_{J^{*}}B_{n_{i}}(x)m^{*}(x)\Myd x \\
=& \sqrt{n}\sum_{i=1}^{n}\xi^{*}_{i,t}\int_{J^{*}}\frac{1}{nh_{n}^{d}}K\left(\frac{x-X_{i}}{h_{n}}\right)  \frac{m^{*}(x)}{\hat{m}(x)} \Myd x \\
=&\sqrt{n}\sum_{i=1}^{n}\xi^{*}_{i,t}
\int_{J^{*}}\frac{1}{nh_{n}^{d}}K\left(\frac{x-X_{i}}{h_{n}}\right)\frac{m^{*}(x)}{m(x)}\Myd x\\
&+\sqrt{n}\sum_{i=1}^{n}\xi^{*}_{i,t}
\int_{J^{*}}\frac{1}{nh_{n}^{d}}K\left(\frac{x-X_{i}}{h_{n}}\right)
\frac{m^{*}(x)}{[m(x)]^{2}}[m(x)-\hat{m}(x)]\Myd x\\
&+\sqrt{n}\sum_{i=1}^{n}\xi^{*}_{i,t}
\int_{J^{*}}\frac{1}{nh_{n}^{d}}K\left(\frac{x-X_{i}}{h_{n}}\right)
\frac{m^{*}(x)}{[\tilde{m}(x)]^{3}}[\hat{m}(x)-m(x)]^{2}\Myd x\\
=& \text{Term I.a}+\text{Term I.b}+\text{Term I.c}
\end{align*}
where $\tilde{m}(x)$ is between $m(x)$ and $\hat{m}(x)$.
The last term is asymptotically uniformly negligible because
\begin{align*}
\sup_{t\in[0,\zeta^{*}]}|\xi^{*}_{i,t}|\leq \frac{1}{v^{*}}\left(1+\frac{1}{v^{*}}\right)
\end{align*}
by Assumption~\ref{SP3},
and
\begin{align*}
&|\text{Term I.c}|\\
\leq&
\left(\frac{2}{u_{0}}\right)^{2}\sup_{x\in J^{*}}|\hat{m}(x)-m(x)|^{2}
\frac{1}{n^{1/2}}\sum_{i=1}^{n}|\xi^{*}_{i,t}|
\int \left|K\left(z\right)\right|m^{*}(X_{i}+h_{n}z)\Myd z\\
=&\Bigo{n^{1/2}\sup_{x\in J^{*}}|\hat{m}(x)-m(x)|^{2}}\\
=&\Smallop{1}
\end{align*}
by Assumptions~\ref{SP2}, \ref{K1}-\ref{K3} and \ref{B1}-\ref{B3}.
In addition, Assumptions~\ref{SM1}, \ref{SM4}, and~\ref{K1}-\ref{K3} imply that the first term satisfies
\begin{align}\label{A0}
\text{Term I.a}
=\frac{1}{\sqrt{n}}\sum_{i=1}^{n}\xi^{*}_{i,t}\frac{m^{*}(X_{i})}{m(X_{i})}+\Smallop{1}.
\end{align}

It suffices to show that the second term
\begin{align}\label{A2}
\text{Term I.b}
\equiv \frac{1}{\sqrt{n}}\sum_{i=1}^{n}\xi^{*}_{i,t}\int_{J^{*}}\frac{1}{h_{n}^{d}}K\left(\frac{x-X_{i}}{h_{n}}\right)
\frac{m^{*}(x)}{[m(x)]^{2}}[m(x)-\hat{m}(x)]\Myd x
\end{align}
is asymptotically uniformly negligible.
Let $u(x)\equiv m^{*}(x)/m(x)$ and $v(x)\equiv m^{*}(x)/(m(x))^{2}$.
Assumptions~\ref{SM1}, \ref{SM4}, and~\ref{K1}-\ref{K3} imply that
\begin{align}\label{A1}
&\int_{J^{*}}\frac{1}{h_{n}^{d}}K\left(\frac{x-X_{i}}{h_{n}}\right)
\frac{m^{*}(x)}{[m(x)]^{2}}[m(x)-\hat{m}(x)]\Myd x\notag\\
=&\int_{J^{*}}\frac{1}{h_{n}^{d}}K\left(\frac{x-X_{i}}{h_{n}}\right)u(x) \Myd x
-\int_{J^{*}}\frac{1}{h_{n}^{d}}K\left(\frac{x-X_{i}}{h_{n}}\right)v(x)\hat{m}(x)\Myd x\notag  \\
=&u(X_{i})+\Bigop{h_{n}^{r}}
-\frac{1}{nh_{n}^{d}}\sum_{j=1}^{n}\int K\left(w\right)v(X_{i}+h_{n}w)
K\left(\frac{X_{i}-X_{j}}{h_{n}}+w\right)\Myd w\notag  \\
=&u(X_{i})-\frac{1}{nh_{n}^{d}}\sum_{j=1}^{n}v(X_{i})K\left(\frac{X_{i}-X_{j}}{h_{n}}\right)
-\frac{1}{n}\sum_{j=1}^{n}Q(X_{i},X_{j};h_{n})+\Smallop{n^{-1/2}},
\end{align}
where
\begin{align*}
&Q(x_{1},x_{2};h_{n})\\
=&\frac{1}{h_{n}^{d}}\int K(z)\left[v(x_{1}+h_{n}z)K\left(\frac{x_{1}-x_{2}}{h_{n}}+z\right)
-v(x_{1})K\left(\frac{x_{1}-x_{2}}{h_{n}}\right)\right]\Myd z.
\end{align*}
Let $\bar{Q}(x;h_{n})=\Exp\left[Q(x,X;h_{n})\right]$.
Substituting (\ref{A1}) into (\ref{A2}) yields
\begin{align}\label{A4}
&\text{Term I.b}\notag\\
= & \frac{1}{n^{3/2}}\sum_{i=1}^{n}\sum_{j=1}^{n}\xi^{*}_{i,t}
\left[u(X_{i})-\frac{1}{h_{n}^{d}}v(X_{i})K\left(\frac{X_{i}-X_{j}}{h_{n}}\right) \right]\notag\\
& -\frac{1}{n^{1/2}}\sum_{i=1}^{n}\xi^{*}_{i,t}\bar{Q}(X_{i};h_{n})
-\frac{1}{n^{3/2}}\sum_{i=1}^{n}\sum_{j=1}^{n}\xi^{*}_{i,t}\left[Q(X_{i},X_{j};h_{n})-\bar{Q}(X_{i};h_{n})\right]
+\Smallop{1}\notag\\
=&\text{Term I.b1}-\text{Term I.b2}-\text{Term I.b3}+\Smallop{1}.
\end{align}
For simplicity, we introduce further notation.
Let $\bar{m}(x;h_{n})\equiv\mathbb{E}\left[\frac{1}{h_{n}^{d}}K\left(\frac{x-X}{h_{n}}\right)\right]$.
For each $t\in[0,\zeta^{*}]$, let
\begin{align}\label{Lfun}
\mathcal{L}(W_{i},W_{j};t,h_{n})\equiv\xi^{*}_{i,t}v(X_{i})
\left[h_{n}^{d}\bar{m}(X_{i};h_{n})-K\left(\frac{X_{i}-X_{j}}{h_{n}}\right) \right],
\end{align}
and
\begin{align}\label{Mfun}
\mathcal{M}(W_{i},W_{j};t,h_{n})\equiv\xi^{*}_{i,t}h_{n}^{d}\left[Q(X_{i},X_{j};h_{n})-\bar{Q}(X_{i};h_{n})\right]
\end{align}
where $W_{i}=(Y_{i},\delta_{i},X_{i})^{\top}$ for each $i\in\{1,2,\ldots,n\}$.
We have
\begin{align}\label{A5}
 & \frac{1}{n^{3/2}}\sum_{i=1}^{n}\sum_{j=1}^{n}\xi^{*}_{i,t}
\left[u(X_{i})-\frac{1}{h_{n}^{d}}v(X_{i})K\left(\frac{X_{i}-X_{j}}{h_{n}}\right) \right]\notag\\
= & \frac{1}{h_{n}^{d}}\frac{1}{n^{3/2}}\sum_{i=1}^{n}\sum_{j=1}^{n}\mathcal{L}(W_{i},W_{j};t,h_{n})
+ \frac{1}{n^{1/2}}\sum_{i=1}^{n}\xi^{*}_{i,t}v(X_{i})\left[m(X_{i})-\bar{m}(X_{i};h_{n})\right].
\end{align}
Since $\mathcal{L}$ is uniformly bounded and $n^{1/2}h_{n}^{d}\rightarrow \infty$ by Assumptions~\ref{B1} and~\ref{B2}, we have
\begin{align*}
\sup_{t\in [0,\zeta^{*}]}
\left|\frac{1}{h_{n}^{d}}\frac{1}{n^{3/2}}\sum_{i=1}^{n}\mathcal{L}(W_{i},W_{i};t,h_{n})\right|
=\Smallop{1}.
\end{align*}
In addition, for each $t\in[0,\zeta^{*}]$ and $h>0$,
$\frac{1}{n(n-1)}\sum_{i=1}^{n}\sum_{j\neq i}\mathcal{L}(W_{i},W_{j},t,h)$
is a degenerate U statistic of order 2 because $\Exp[\mathcal{L}(W_{i},W_{j},t,h)|W_{i}]=\Exp[\mathcal{L}(W_{i},W_{j},t,h)|W_{j}]=0$ with probability one.
Lemma~\ref{Euclidean} shows that $\{\mathcal{L}(\cdot,\cdot;t,h):t\in [0,\zeta^{*}],h>0\}$ is an Euclidean class.
Corollary 4 of \citet{Sherman1994} implies that
\begin{align*}
\sup_{(t,h)\in [0,\zeta^{*}]\times(0,\infty)}\Bigg|\frac{1}{n}\sum_{i=1}^{n}\sum_{j\neq i}
\mathcal{L}(W_{i},W_{j};t,h)\Bigg|
=\Bigop{1}.
\end{align*}
It follows that
\begin{align}\label{A6}
&\sup_{t\in [0,\zeta^{*}]}\Bigg|
\frac{1}{h_{n}^{d}}\frac{1}{n^{3/2}}\sum_{i=1}^{n}\sum_{j=1}^{n}\mathcal{L}(W_{i},W_{j};t,h_{n})\Bigg|\notag\\
\leq &\frac{1}{n^{3/2}h_{n}^{d}}
\left(
\sup_{(t,h)\in [0,\zeta^{*}]\times(0,\infty)}\Bigg|\sum_{i=1}^{n}\sum_{j\neq i}
\mathcal{L}(W_{i},W_{j};t,h)\Bigg|+\sup_{t\in [0,\zeta^{*}]}
\left|\sum_{i=1}^{n}\mathcal{L}(W_{i},W_{i};t,h_{n})\right|
\right)
\notag\\
=&\Smallop{1}.
\end{align}
In addition, since $\bar{m}(x;h_{n})=m(x)+\Bigop{h_{n}^r}$ by kernel smoothing techniques, we have
\begin{align}\label{A7}
&\sup_{t\in [0,\zeta^{*}]}
\Bigg|\frac{1}{n^{1/2}}\sum_{i=1}^{n}\xi^{*}_{i,t}v(X_{i})\left[m(X_{i})-\bar{m}(X_{i};h_{n})\right]\Bigg|\notag\\
\leq &\frac{1}{n^{1/2}}\sum_{i=1}^{n}\sup_{t\in [0,\zeta^{*}]}\Big|\xi^{*}_{i,t}v(X_{i})\Big|
\Big|m(X_{i})-\bar{m}(X_{i};h_{n})\Big|\notag\\
= &\Bigop{n^{1/2}h_{n}^{r}}\notag\\
= &\Smallop{1}
\end{align}
by Assumption~\ref{B3}.
Substituting~(\ref{A6}) and~(\ref{A7}) into~(\ref{A5}), we obtain that Term I.b1 is asymptotically uniformly negligible.
In addition, after some simple but tedious calculation,
we show that under Assumptions~\ref{K1}-\ref{K2} and~\ref{SM1},
$\bar{Q}(X_{i};h_{n})=\Bigop{h_{n}^{r}}$.
It follows that Term I.b2 is asymptotically uniformly negligible; that is,
\begin{align}\label{A8}
\sup_{t\in [0,\zeta^{*}]}\Bigg|\frac{1}{n^{1/2}}\sum_{i=1}^{n}\xi^{*}_{i,t}\bar{Q}(X_{i};h_{n})\Bigg|
\leq &\Bigop{n^{1/2}h_{n}^{r}}=\Smallop{1}.
\end{align}

It remains to show that Term I.b3 is asymptotically uniformly negligible.
For each $t\in[0,\zeta^{*}]$ and $h>0$, $\frac{1}{n(n-1)}\sum_{i=1}^{n}\sum_{j\neq i}\mathcal{M}(W_{i},W_{j};t,h)$ is a degenerate U statistic of order 2 because $\Exp[\mathcal{M}(W_{i},W_{j},t,h)|W_{i}]=\Exp[\mathcal{M}(W_{i},W_{j},t,h)|W_{j}]=0$ with probability one.
Lemma~\ref{Euclidean} shows that $\{\mathcal{M}(\cdot,\cdot;t,h):t\in [0,\zeta^{*}],h\in(0,1)\}$ is an Euclidean class; thus, we have
\begin{align*}
\sup_{(t,h)\in [0,\zeta^{*}]\times(0,1)}
\Bigg|\frac{1}{n}\sum_{i=1}^{n}\sum_{j\neq i}\mathcal{M}(W_{i},W_{j};t,h)\Bigg|
=\Bigop{1}
\end{align*}
by Corollary 4 of \citet{Sherman1994}.
We also have
\begin{align*}
\sup_{t\in [0,\zeta^{*}]}
\left|\frac{1}{n^{3/2}h_{n}^{d}}\sum_{i=1}^{n}\mathcal{M}(W_{i},W_{i};t,h_{n})\right|
=\Smallop{1}
\end{align*}
because $\mathcal{M}$ is uniformly bounded and $n^{1/2}h_{n}^{d}\rightarrow \infty$ by Assumptions~\ref{B1} and~\ref{B2}.
It follows that Term I.b3 is asymptotically uniformly negligible; that is,
\begin{align}\label{A9}
&\sup_{t\in [0,\zeta^{*}]}\Bigg|\frac{1}{n^{3/2}}\sum_{i=1}^{n}\sum_{j=1}^{n}
\xi^{*}_{i,t}\left[Q(X_{i},X_{j};h_{n})-\bar{Q}(X_{i};h_{n})\right]\Bigg|\notag\\
\leq &\frac{1}{n^{3/2}h_{n}^{d}}\left(\sup_{(t,h)\in [0,\zeta^{*}]\times(0,1)}
\Bigg|\sum_{i=1}^{n}\sum_{j\neq i}\mathcal{M}(W_{i},W_{j};t,h)\Bigg|
+\sup_{t\in [0,\zeta^{*}]}
\left|\sum_{i=1}^{n}\mathcal{M}(W_{i},W_{i};t,h_{n})\right|\right)
\notag\\
=&\Smallop{1}.
\end{align}
Combining~(\ref{A4})-(\ref{A9}), we obtain that \text{Term I.b} is of order $\Smallop{1}$ uniformly in $t\in[0,\zeta^{*}]$.
Hence, we obtain
\begin{align*}
\text{Term I}
= \frac{1}{\sqrt{n}}\sum_{i=1}^{n}\xi^{*}_{i,t}\frac{m^{*}(X_{i})}{m(X_{i})}+\Smallop{1}
\end{align*}
where the remainder term is asymptotically uniformly negligible in $t\in[0,\zeta^{*}]$.

Following similar arguments, we can also show that Term II is asymptotically uniformly negligible.

\end{proof}

\begin{Lem}\label{Euclidean}
Suppose that the assumptions of Proposition~\ref{PRO3} hold.
\begin{enumerate}[(i)]
\item Let $\xi$ be the function defined in Proposition~\ref{PRO2}.
The class of functions
\begin{align*}
\{(y,\delta)\mapsto\xi(y,\delta;t):t\in [0,\zeta^{*}]\}
\end{align*}
is Euclidean for some envelope.
\item Let $\xi^{*}$ be the function defined in Proposition~\ref{PRO3}.
The class of functions
\begin{align*}
\{(y,\delta,x)\mapsto\xi^{*}(y,\delta;t,x):t\in [0,\zeta^{*}]\}
\end{align*}
is Euclidean for some envelope.
\item Let $\mathcal{L}$ be the function defined in~(\ref{Lfun}).
The class of functions
\begin{align*}
\{(w_{1},w_{2})\mapsto\mathcal{L}(w_{1},w_{2};t,h):t\in [0,\zeta^{*}], h>0\}
\end{align*}
is Euclidean for some envelope.
\item Let $\mathcal{M}$ be the function defined in~(\ref{Mfun}).
The class of functions
\begin{align*}
\{(w_{1},w_{2})\mapsto\mathcal{M}(w_{1},w_{2};t,h):t\in [0,\zeta^{*}], h\in(0,1)\}
\end{align*}
is Euclidean for some envelope.
\end{enumerate}
\end{Lem}
\begin{proof}
(\textit{i})
For any $t,t'\in[0,\xi^{*}]$, we have
\begin{align*}
\left|\int_{0}^{t}\frac{F^{\delta}_{Y}(\!\!\Myd s)}{(1-F_{Y}(s))^{2}}-\int_{0}^{t'}\frac{ F^{\delta}_{Y}(\!\!\Myd s)}{(1-F_{Y}(s))^{2}}\right|
\leq & \left(\frac{1}{1-F_{Y}(\zeta^{*})}\right)^{2}\left|F^{\delta}_{Y}(t)-F^{\delta}_{Y}(t')\right|\\
\leq & \left(\frac{1}{1-F_{Y}(\zeta^{*})}\right)^{2}\sup_{s\in[0,\zeta^{*}]}\left|\frac{\partial F^{\delta}_{Y}(s)}{\partial t}\right||t-t'|\\
\leq & \left(\frac{1}{1-F_{Y}(\zeta^{*})}\right)^{2}
\left(\sup_{s\in[0,\zeta^{*}]}f_{T}(s)\right)|t-t'|
\end{align*}
by Assumption~\ref{I1}.
Lemmas 2.13 and 2.14 of \citet{PakesPollard1989} imply that the class
\begin{align*}
\left\{(y,\delta)\mapsto \int_{0}^{y\wedge t}\frac{ F^{\delta}_{Y}(\!\!\Myd s)}{(1-F_{Y}(s))^{2}}:t\in[0,\zeta^{*}]\right\}
\end{align*}
is Euclidean for the constant envelope $\left[1+2\zeta^{*}\sup_{s\in[0,\zeta^{*}]}f_{T}(s)\right]/[1-F_{Y}(\zeta^{*})]^{2}$ under Assumption~\ref{SM5} because
\begin{align*}
\int_{0}^{y\wedge t}\frac{F^{\delta}_{Y}(\!\!\Myd s)}{(1-F_{Y}(s))^{2}}=\min\left\{\int_{0}^{y}\frac{F^{\delta}_{Y}(\!\!\Myd s)}{(1-F_{Y}(s))^{2}},\int_{0}^{t}\frac{F^{\delta}_{Y}(\!\!\Myd s)}{(1-F_{Y}(s))^{2}}\right\}.
\end{align*}
In addition, Lemma 19.15 of \citet{Vaart1998} implies that the class $\{y\mapsto\Ind{[y\leq t]}:t\in[0,\zeta^{*}]\}$ is Euclidean for a constant envelope because $\{(-\infty,t]:t\in\mathbb{R}\}$ is a VC class.
Note that we also have
\begin{align*}
|F_{T}(t)-F_{T}(t')|\leq \left(\sup_{s\in[0,\zeta^{*}]}f_{T}(s)\right)|t-t'|
\end{align*}
for any $t,t'\in[0,\xi^{*}]$.
Applying Lemmas 2.13 and 2.14 of \citet{PakesPollard1989} again yields that the class $\{(y,\delta)\mapsto\xi(y,\delta;t):t\in [0,\zeta^{*}]\}\}$ is Euclidean for some envelope.\\[0.3cm]
(\textit{ii})
For any $t,t'\in[0,\xi^{*}]$ and $x\in J$, we have
\begin{align*}
|F_{T|X}(t|x)-F_{T|X}(t'|x)|\leq \left(\sup_{s\in[0,\zeta^{*}]}f_{T|X}(s|x)\right)|t-t'|
\end{align*}
and
\begin{align*}
\left|\int_{0}^{t}\frac{F^{\delta}_{Y|X}(\!\!\Myd s|x)}{(1-F_{Y|X}(s|x))^{2}}-\int_{0}^{t'}\frac{ F^{\delta}_{Y|X}(\!\!\Myd s|x)}{(1-F_{Y|X}(s|x))^{2}}\right|
\leq & \left(\frac{1}{v^{*}}\right)^{2}\left|F^{\delta}_{Y|X}(t|x)-F^{\delta}_{Y|X}(t'|x)\right|\\
\leq & \left(\frac{1}{v^{*}}\right)^{2}
\sup_{s\in[0,\zeta^{*}]}\left|\frac{\partial F^{\delta}_{Y|X}(s|x)}{\partial t}\right||t-t'|\\
\leq & \left(\frac{1}{v^{*}}\right)^{2}
\left(\sup_{s\in[0,\zeta^{*}]}f_{T|X}(s|x)\right)|t-t'|
\end{align*}
by Assumption~\ref{I2}.
From Lemmas 2.13 and 2.14 of \citet{PakesPollard1989} and under Assumption~\ref{SM5}, the class $\{x\mapsto 1-F_{T|X}(t|x):t\in[0,\zeta^{*}]\}$ is Euclidean for the envelope $1+2\zeta^{*}\sup_{s\in[0,\zeta^{*}]}f_{T|X}(s|x)$, and the class
\begin{align*}
\left\{(y,\delta,x)\mapsto \int_{0}^{y\wedge t}\frac{ F^{\delta}_{Y|X}(\!\!\Myd s|x)}{(1-F_{Y|X}(s|x))^{2}}:t\in[0,\zeta^{*}]\right\}
\end{align*}
is Euclidean for the envelope $\left[1+2\zeta^{*}\sup_{s\in[0,\zeta^{*}]}f_{T|X}(s|x)\right]/(v^{*})^{2}$ because
\begin{align*}
\int_{0}^{y\wedge t}\frac{F^{\delta}_{Y|X}(\!\!\Myd s|x)}{(1-F_{Y|X}(s|x))^{2}}=\min\left\{\int_{0}^{y}\frac{F^{\delta}_{Y|X}(\!\!\Myd s|x)}{(1-F_{Y|X}(s|x))^{2}},\int_{0}^{t}\frac{F^{\delta}_{Y|X}(\!\!\Myd s|x)}{(1-F_{Y|X}(s|x))^{2}}\right\}.
\end{align*}
Since $\{(-\infty,t]:t\in\mathbb{R}\}$ is a VC class, Lemma 19.15 of \citet{Vaart1998} implies that the class $\{(y,\delta,x)\mapsto\frac{\Ind{[y\leq t]}\delta}{1-F_{Y|X}(y|x)}:t\in[0,\zeta^{*}]\}$ is Euclidean for a constant envelope.
Therefore, the class of functions $\{(y,x,\delta)\mapsto\xi^{*}(y,\delta;t,x):t\in [0,\zeta^{*}]\}$ is Euclidean for some envelope by Lemma 2.14 of \citet{PakesPollard1989}.\\[0.3cm]
(\textit{iii})
Let $g(x_{1},x_{2})=x_{1}-x_{2}$ and $\mathcal{K}=\{u\mapsto K(u/h):h>0\}$.
Assumption~\ref{K1} implies that $\mathcal{K}$ is a VC-subgraph class.
(See the discussion on page 911 of \citet{GineGuillou2002}.)
It follows from Lemma 2.6.18(vii) of \citet{VaartWellner1996} that the class $\mathcal{F}\circ g=\{(x_{1},x_{2})\mapsto K\left((x_{1}-x_{2})/h\right):h>0\}$ is a VC-subgraph class; thus $\mathcal{F}\circ g$ is Euclidean for a constant envelope.
Lemma 5 of \citet{Sherman1994} implies that the class $\{x\mapsto h^{d}\bar{m}(x;h):h>0\}$ is also Euclidean for a constant envelope. 
Since the class $\{(y,\delta,x)\mapsto\xi^{*}(y,\delta;t,x):t\in [0,\zeta^{*}]\}$ is Euclidean by part (\textit{ii}), the class $\{(w_{1},w_{2})\mapsto\mathcal{L}(w_{1},w_{2};t,h):t\in [0,\zeta^{*}], h>0\}$ is Euclidean for some envelope by Lemma 2.14 of \citet{PakesPollard1989}.\\[0.3cm]
(\textit{iv})
By part \textit{(ii)}, Lemma 5 of \citet{Sherman1994}, and Lemma 2.14 of \citet{PakesPollard1989}, it suffices to show that $\left\{(x_{1},x_{2})\mapsto h^{d}Q(x_{1},x_{2};h):h\in(0,1)\right\}$ is Euclidean.
Since there is a constant $c$ such that for $h_{1},h_{2}\in(0,1)$,
\begin{align*}
\sup_{x,z\in J}|v(x+h_{1}z)-v(x+h_{2}z)|\leq c|h_{1}-h_{2}|
\end{align*}
by Assumptions~\ref{SM1} and~\ref{SM4},
the class $\left\{(x,z)\mapsto v(x+hz):h\in(0,1)\right\}$ is Euclidean by Lemma 2.13 of \citet{PakesPollard1989}.
In addition, Lemma 22(i) of \citet{NolanPollard1987} imply that the class
\begin{align*}
\left\{(x_{1},x_{2},z)\mapsto K\left(\frac{x_{1}-x_{2}}{h}+z\right):h\in(0,1)\right\}
\end{align*}
is Euclidean by Assumptions~\ref{K1} and~\ref{K4}.
Let $\mathscr{U}$ be the measure on $[-1,1]^{d}$ associated with a uniform random variable on $[-1,1]^{d}$. Note that
\begin{align*}
&h^{d}Q(x_{1},x_{2})\\
=&\int 2^{d}K(z)\left[
v(x_{1}+hz)K\left(\frac{x_{1}-x_{2}}{h}+z\right)-v(x_{1})K\left(\frac{x_{1}-x_{2}}{h}\right)
\right]\mathscr{U}(\!\!\Myd z).
\end{align*}
It follows from Lemma 2.14 of \citet{PakesPollard1989} and Lemma A.2. of \citet{GhosalSenEtAl2000} that the class
$\left\{(x_{1},x_{2})\mapsto h^{d}Q(x_{1},x_{2}):h\in(0,1)\right\}$ is Euclidean.
\end{proof}

\end{appendices}

\normalsize
\bibliography{Mybib_RCPE}
\bibliographystyle{ecta}

\end{document}